%% file: main.tex
\documentclass{llncs}
\usepackage{packages}
\usepackage{pgfplots}
\usepackage{pgfplotstable}

\usepackage{macros}
\usepackage{cite}
\usepackage{parskip}
\title{RYDE: A Digital Signature Scheme based on Rank-Syndrome-Decoding Problem with MPCitH Paradigm}

\author{
    Loïc Bidoux\inst{1} \and Jesús-Javier Chi-Domínguez\inst{1} \and Thibauld Feneuil\inst{2 \and 3} \and \\
    Philippe Gaborit\inst{4} \and Antoine Joux\inst{5} \and Matthieu Rivain\inst{3} \and Adrien Vinçotte\inst{4}
}

\institute{
    Technology Innovation Institute, UAE \and
    Sorbonne Université, CNRS, INRIA, Institut de Mathématiques\\ de Jussieu-Paris Rive Gauche, Ouragan, Paris, France \and
    CryptoExperts, Paris, France \and
    University of Limoges, France \and
    CISPA, Helmholtz Center for Information Security, Saarbrücken
}

\begin{document}

\let\oldaddcontentsline\addcontentsline
\def\addcontentsline#1#2#3{}
\maketitle
\def\addcontentsline#1#2#3{\oldaddcontentsline{#1}{#2}{#3}}

\begin{abstract}
    We present a signature scheme based on the Syndrome-Decoding problem in rank metric. It is a construction from multi-party computation (MPC), using a MPC protocol which is a slight improvement of the linearized-polynomial protocol used in \cite{F22}, allowing to obtain a zero-knowledge proof thanks to the MPCitH paradigm. We design two different zero-knowledge proofs exploiting this paradigm: the first, which reaches the lower communication costs, relies on additive secret sharing and uses the hypercube technique \cite{AGHHJY22}; and the second relies on low-threshold linear secret sharing as proposed in \cite{FR22}. These proofs of knowledge are transformed into signature schemes thanks to the Fiat-Shamir heuristic \cite{FS}, allows us to obtain signatures above 6kB. These performances prompted us to propose this signature to the Post-Quantum Cryptography Standardization process organized by NIST.
\end{abstract}

\newpage

\section{Introduction}

Zero-Knowledge (ZK) Proofs of Knowledge (PoK) have become significant cryptographic primitives thanks to their various applications. Such protocols allow a party (the prover) to convince an other one (the verifier) that he knows a secret information without reveal anything about it. ZK proofs are useful in many contexts: they allow to construct identification schemes, they can be turned in signature schemes thanks to the Fiat-Shamir transform \cite{FS} or the Unruh transform \cite{Unr}. In recent years, many signatures have been constructed using this approach, for example from Zero-knowledge proofs of knowledge of the solution to an instance of the Syndrome Decoding problem \cite{Stern93}, the Multivariate Quadratic problem \cite{SSH11}, or the Permuted Kernel Problem \cite{SHA89}.

We propose here a code based signature scheme, whose security relies on the difficulty to solve the syndrome decoding problem: for a random matrix $\bm{H}\in\Fqm^{(n-k)\times n}$ and a vector $\bm{y}\in\Fqm^{n-k}$, this problem asks to find a vector $\bm{x}\in\Fqm^n$ of small weight such that $\bm{H}\bm{x}^T=\bm{y}$. This problem is known to be NP-hard and accepted as secured for reasonable parameters.\\
We consider here the rank weight of $\bm{x}$, which is equal to the dimension of the $\Fq$-linear subspace generated by its coefficients in $\Fqm$. The notion of error in this metric is completely different from that of Hamming: counting the number of non-zero coordinates induces a local notion of error, while determining the dimension of a vector subspace generated by a set of coordinates implies having a general vision of it. Beyond their use in cryptography, rank metric codes are therefore effective in practice for correcting errors in the form of blocks.\\
A first proposition of PoK for Syndrome Decoding problem has be done by \cite{Stern93} in Hamming metric, but its soundness error, which corresponds to the probability that a malicious prover convince the verifier that he knows the witness of the Syndrome Decoding instance, is equal to $2/3$. Although by repeating the protocol $\tau$ times the soundness equal to $(2/3)^\tau$ can be arbitrary close to 0, it results from the high amount of information to transmit a too long signature. His approach has been generalized within the framework of the rank metric, and has given rise to numerous protocol proposals, for example \cite{GSZ11, FJR21, BG22}. Among these, some of which have a soundness equal to $1/N$ for some parameter $N$ (the quantity of information to be transmitted increases with N, but with a better trade-off).

A possible technique for design efficient Proofs of Knowledge is to use the MPC-in-the-Head paradigm, introduced in \cite{IKOS07}, that relies on symmetric primitives. The idea which will be improved by \cite{KKW18}, by using MPC in their pre-processing model. This technique relies on the multi-party computation: the prover emulates "in his head" an additive MPC protocol where all information are split in $N$ shares. It commits all the emulated parties and reveals the views of $N-1$ among them. The only way for a malicious prover to cheat would be to do so on the unopened part only, which happen with probability $1/N$. More recently, \cite{AGHHJY22} suggests structuring the additive shares in a hypercube to improve the performances of MPC protocols. It is also possible to construct signatures from threshold multi-party computation \cite{FJR22}, which allows a gain in the speed of signature verification at the cost of bigger signatures.\\
This framework has the advantage of being very general and can be adapted to prove knowledge of the witness of an instance for any type of problem. For example, the MPC-in-the-Head paradigm has been used to construct several signatures, relying on various problems, submitted to the NIST for their post-quantum standardization process: SDitH \cite{sdith} relies on the Syndrome Decoding problem in Hamming metric, MIRA \cite{mira} and MiRitH \cite{mirith} relies on the Min-Rank problem, MQOM \cite{mqom} and Biscuit \cite{biscuit} relies on the Multivariate Quadratic problem (or a structured variant), and PERK \cite{perk} relies on the Permuted Kernel problem.

\textbf{Contributions.} We propose two complete signature schemes which adapt protocol for Rank-Syndrome-Decoding using linearized polynomials described in \cite{F22} to the methods of hypercube \cite{AGHHJY22} and threshold MPC \cite{FJR22}, whose adaptation isn't straightforward. The performance of the first of these two signatures surpasses those of all the others relying on Rank Syndrome Decoding problem mentioned above: the hypercube optimization allows us to obtain an efficient quantum-resistant signature whose size is below 6kB. These results prompted us to propose this signature to the Post-Quantum Cryptography Standardization process organized by NIST \cite{ryde}\footnote[1]{All proposed signatures can be found here: https://csrc.nist.gov/projects/pqc-dig-sig/round-1-additional-signatures}. The use of the threshold-MPC in the second signature allows a gain in the speed of the verification algorithm. We also provide complete security proofs of these two signature schemes.

\textbf{Paper organization.} In Section 2, we begin by defining all cryptographic and mathematical notions necessary for understanding. In Section 3, we give a high-level overview of the signature scheme (and leave low-level instructions to the specifications document). In Section 4, we describe the chosen parameters and their associated performances. In Section 5, we detail security proofs of the schemes. In Section 6, we detail attacks against signatures designed with the Fiat-Shamir transform and attacks against the Rank-Syndrome-Decoding problem.

\section{Preliminary notions}

\subsection{Notations and conventions}

Let $A$ a randomized algorithm. We write $y\leftarrow A(x)$ the output of the algorithm for the input $x$. If $S$ is a set, we write $x\sampler S$ the uniform sampling of a random element $x$ in $S$. We write $x\samples{s} S$ the pseudo-random sampling in $S$ with seed $s$.

We denote by $\mathbb{F}_q$ the finite field of order $q$. We use bold letters to denote vectors and matrices (for example, $\bm{u}\in\mathbb{F}_q^n$ and $u\in\mathbb{F}_q$).

A function $\mu:\mathbb{N}\rightarrow\mathbb{R}$ is said \textit{negligible} if, for every positive polynomial $p(\cdot)$, there exists an integer $N_{p} > 0$ such that for every $\lambda > N_{p}$, we have $|\mu(\lambda)| < 1/{p(\lambda)}$. When not made explicit, a negligible function in $\lambda$ is denoted $\mathsf{negl}(\lambda)$ while a polynomial function in $\lambda$ is denoted $\mathsf{poly}(\lambda)$. We further use the notation $\mathsf{poly}(\lambda_1, \lambda_2, ...)$ for a polynomial function in several variables.

Two distributions $\{D_\lambda\}_\lambda$ and $\{E_\lambda\}_\lambda$ indexed by a security parameter $\lambda$ are $(t,\varepsilon)$-\textit{indistinguishable} (where $t$ and $\varepsilon$ are $\mathbb{N} \to \mathbb{R}$ functions) if, for any algorithm $\mathcal{A}$ running in time at most $t(\lambda)$ we have 
$$\big| \Pr[\mathcal{A}^{D_\lambda}()=1] - \Pr[\mathcal{A}^{E_\lambda}()=1 ]\big| \leq \varepsilon(\lambda)~, $$ 
with $\mathcal{A}^{Dist}$ meaning that $\mathcal{A}$ has access to a sampling oracle of distribution $Dist$.
The two distributions are said
\begin{itemize}
    \item \textit{computationally indistinguishable} if $\varepsilon \in \mathsf{negl}(\lambda)$ for every $t \in \mathsf{poly}(\lambda)$;
    \item \textit{statistically indistinguishable} if $\varepsilon \in \mathsf{negl}(\lambda)$  for every (unbounded) $t$;
    \item \textit{perfectly indistinguishable} if $\varepsilon = 0$ for every (unbounded) $t$.
\end{itemize}

\subsection{Security notions}

\subsubsection{Digital signature schemes}

We begin by remembering basic definitions and notions of security about signature schemes.

\begin{definition}[Digital signature scheme]
    A digital signature scheme $\mathsf{DSS}$ with security parameter $\lambda$ is a triplet of polynomial time algorithms $(\mathsf{KeyGen}, \mathsf{Sign}, \mathsf{Verif})$ such that:\begin{itemize}
        \item The key-generation algorithm $\mathsf{KeyGen}$ is a probabilistic algorithm which outputs a pair of keys $(\mathsf{pk},\mathsf{sk})$;
        \item The signing algorithm $\mathsf{Sign}$, possibly probabilistic, which takes as inputs a message $m$ to be signed and the secret key $\mathsf{sk}$, and outputs a signature $\sigma$;
        \item The verification algorithm $\mathsf{Verif}$ which takes as inputs the public key $\mathsf{pk}$, the message $m$ and its signature $\sigma$, and outputs a bit $b$. Output $1$ indicate that the signature is considered as valid.
    \end{itemize}
\end{definition}

A correct signature scheme verifies the following property: if $(\pk,\sk)\leftarrow\mathsf{KeyGen}$, for all messages $m$ signed by $\sigma\leftarrow\mathsf{Sign}(\pk,m)$, we have $1\leftarrow\mathsf{Verif}(\sk,m,\sigma)$. This means that a signature which is correctly generated is always accepted (when using the right keys).

The standard security notion for digital signature schemes is existential unforgeability under adaptive chosen message attacks (EUF-CMA) is defined as follows: 
\begin{definition}[EUF-CMA]
    We can define the following game $G_{\textsc{euf-cma}}$ where ${\mathsf{Sign}(\sk,\cdot)}$ is an oracle that sign a message $m^*$: \\
    $\textsl{(\pk,\sk)} \leftarrow \mathsf{KeyGen}()$ \\
    $(m,\sigma) \leftarrow \mathcal{A}^{Sign(\sk,\cdot)}(\textsf{pk})$ \\
    \line(1,0){100}\\
    returns $1$ if $\textsf{Verif}(m,\sigma,\textsf{pk}) = 1$ and $m$ was not queried to ${\mathsf{Sign}(\sk,\cdot)}$ \\
    The signature scheme is EUF-CMA if, for every polynomial adversary $\mathcal{A}$, $\prb\big(G_{\textsc{euf-cma}}(\mathcal{A}) =1\big)$ is negligible.
\end{definition}

\subsubsection{Commitment schemes}

The security of the signature relies on a commitment scheme, which allows to commit on a value while hiding it until a possible future opening. We consider two properties: the commitment reveals no information about what is committed (hiding property), and there is only one (computationally tractable) way to open the commitment (binding property).

A commitment scheme takes as input any element $x$ (the value to commit) and a random tape $\rho\in\{0,1\}^\lambda$. When the secret value $x$ is revealed, the commitment $\cmt=\Com(x,\rho)$ can be verified by revealing $\rho$, which prove that $x$ has not been modified in the meantime.

\begin{definition}[Hiding]
    A commitment scheme $\Com$ is computationally hiding if, for every $m_0,m_1$, the distributions of $$\big\{ \Com(m_0,\rho), \rho \sampler\{0,1\}^\lambda\big\}\text{ and }\big\{ \Com(m_1,\rho), \rho \sampler\{0,1\}^\lambda\big\}$$ are computationally indistinguishable.
\end{definition}

\begin{definition}[Binding]
    A commitment scheme $Com$ is computationally biding if, for every PPT algorithm $\mathcal{A}$, we have: $$ \prb\big(\Com(m,\rho) = \Com(m',\rho') \cap m \ne m', (m,\rho,m',\rho') \leftarrow \mathcal{A}\big)<\mu(\lambda)$$ where $\mu$ is a negligible function.
\end{definition}

\subsubsection{Pseudo-random generators}

\begin{definition}[Pseudo-random Generators]
Let $G : \{0,1\}^* \rightarrow \{ 0,1\}^*$, $\ell$ a polynomial such that $G(s) \in \{0,1\}^{\ell(\lambda)}$ for $s\in\{0,1\}^\lambda$. G is a $(t,\epsilon)$-secure pseudo-random generator if: \begin{itemize}
    \item $\ell(\lambda) > \lambda$
    \item The distributions $\{ G(s)\vert\: s\leftarrow \{ 0,1\}^\lambda\}$ and $\{r\:\vert\: r \sampler \{ 0,1\}^{\ell(\lambda)}\}$ are $(t,\epsilon)$-indistinguishable
\end{itemize}
\end{definition}

In our constructions, we need \emph{puncturable pseudo-random functions} (puncturable PRF). Such a primitive produces $N$ pseudo-random values and provides a way to reveal all the values but one. A standard construction of puncturable PRFs can be derived from the tree-based construction of \cite{ggm86}, usually denoted as the GGM construction. The idea is to use a tree PRG in which one uses a pseudorandom generator to expand a root seed into $N$ subseeds in a structured way. This can be illustrated quite easily with the following figure:

\begin{figure}[H]
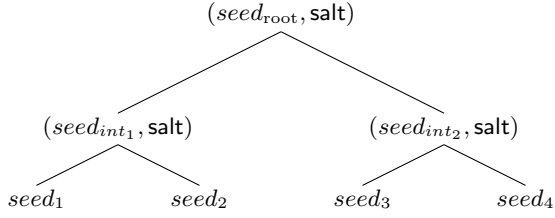

\Tree[.$(seed_{\text{root}},\salt)$ [.$(seed_{int_1},\salt)$ [.$\qquad seed_{1}\qquad$ ]
               [.$\qquad seed_{2}\qquad$ ]]
          [.($seed_{int_2},\salt)$ [.$\qquad seed_{3}\qquad$ ]
               [.$\qquad seed_{4}\qquad$ ]]]
               \captionof{figure}{Example of GGM tree of depth 2}
               \label{TreePRG}
\end{figure}
where each intermediate seed generates two new seeds, and $\salt$ is a random value in $\{0, 1\}^{\lambda}$.

Now, imagine one is looking to reveal $seed_{1}$, $seed_{3}$, and $seed_{4}$, and hide $seed_{2}$. Then, all he has to do is reveal $seed_{int_2}$ and $seed_{1}$ ($\salt$ being known). It is impossible to retrieve $seed_{2}$, as we don't know the previous seed, but it is possible to retrieve the others. In this small example, we took $N=4$. This is especially interesting as, in the additive-based MPCitH transformation we see later, we reveal $N-1$ leaves. This means that we can do this operation by revealing only $\log_2N$ leaves instead of $N-1$. More generally, for a GGM tree with $N$ leaves and given a subset $I \subset \oneto{N}$, it is possible to reveal all the leaves but the ones in $I$ by revealing at most $|I|\log_2(\frac{N}{|I|})$ nodes.

\subsubsection{Merkle trees}

A collision-resistant hash function (that we note $\hash_M$) can be used to build a \textit{Merkle Tree}. Given inputs $v_1 \dots v_N$, we define $\mathsf{Merkle}(v_1 \dots v_N)$ as:  \begin{equation}
    \mathsf{Merkle}(v_1 \dots v_N) =
    \begin{cases}
    \hash_M( {\mathsf{Merkle}(v_1 \dots v_{\frac{N}{2}}) } \| {\mathsf{Merkle}(v_{\frac{N}{2}+1} \dots v_N)}) \text{ if } N>1 \\
    \hash_M(v_1) \text{ if } N=1 \\
    \end{cases}
\end{equation}

A Merkle Tree is a special case of vector commitment scheme allowing efficient opening of a subset of the committed values $v_i$. Given $I \subset \oneto{N}$, it is possible to verify that the $(v_i)_{i\in I}$ were indeed used to build the Merkle Tree only by revealing at most $|I|\log_2\Big(\frac{N}{|I|}\Big)$ hash values. To proceed, we just need to reveal the sibling paths of the $(v_i)_{i\in I}$. Those paths are called the authentication paths and are denoted $\mathsf{auth}((v_1 \dots v_N),I)$.

\subsubsection{Zero-knowledge proofs of knowledge}

We define here the concept of proof of knowledge, following \cite{AFK21} notations. Let $R$ an NP-relation. $(x,\omega)\in R$ is a statement-witness pair where $x$ is the statement and $\omega$ an associated witness. The set of valid witness for a statement $x$ is $R(x)=\{\omega  \mid (x,\omega)\in R\}$. A prover $\mathcal{P}$ wants to use a proof of knowledge to convince a verifier $\mathcal{V}$ that he knows a witness $\omega$ for a statement $x$.

\begin{definition}[Proof of knowledge]
    A proof of knowledge for a relation $R$ with soundness $\epsilon$ is a two-party protocol between a prover $\mathcal{P}$ and a verifier $\mathcal{V}$ with a public statement $x$, where $\mathcal{P}$ wants to convince $\mathcal{V}$ that he knows $\omega$ such that $(x,\omega)\in R$.\\
    Let us denote $\langle \mathcal{P}(x,\omega),\mathcal{V}(x)\rangle$ the transcript between $\mathcal{P}$ and $\mathcal{V}$. A proof of knowledge must be perfectly complete, i.e.: $$\prb\big(\mathsf{Accept}\leftarrow \langle \mathcal{P}(x,\omega),\mathcal{V}(x)\rangle\big) = 1$$
    If there exists a polynomial time prover $\Tilde{\mathcal{P}}$ such that: $$\Tilde{\epsilon}=\prb \big(\mathsf{Accept}\leftarrow \langle \Tilde{\mathcal{P}}(x),\mathcal{V}(x)\rangle\big) >\epsilon$$ then there exists an algorithm $\mathcal{E}$, called extractor, which given rewindable access to $\Tilde{\mathcal{P}}$, outputs a valid witness $\omega' \in R(x)$  in polynomial time in terms of $(\lambda,\frac{1}{\Tilde{\epsilon}-\epsilon})$ with probability at least $\frac{1}{2}$.
\end{definition}

We now introduce the notion of honest-verifier zero-knowledge for a proof of knowledge (PoK):

\begin{definition}[Honest-Verifier Zero-Knowledge]
    A PoK satisfies the Honest-Verifier Zero-Knowledge (HZVK) property if there exists a polynomial time simulator $\mathsf{Sim}$ that given as input a statement $x$ and random challenges $(\ch_1,...,\ch_n)$, outputs a transcript $\lbrace \mathsf{Sim}(x,\ch_1,...,\ch_n),\mathcal{V}(x)\rbrace$ which is computationally indistinguishable from the probability distribution of transcripts of honest executions between a prover $\mathcal{P}(x,w)$ and a verifier $\mathcal{V}(x)$.
\end{definition}

\subsubsection{Fiat-Shamir transform}

The Fiat-Shamir transform is a generic process allowing to transform an HVZK proof of knowledge into a signature scheme. The main adaptation lies in the removal of the interaction in the protocol: one needs to pull the challenges in a pseudo-random (and verifiable) way to sign a message without the assistance of a verifier. Note that the protocol must be repeated several times to achieve the desired level of security. We note $\tau$ the number of repetitions. One see below that there is an effective attack against such signatures and so, $\tau$ must be chosen carefully (not too small).

Let us describe this transformation for the case of a 5-round protocol. The prover begins as in the zero-knowledge proof by committing all the auxiliary information. At the end of the first step, the prover computes: $$h_1=\hash_1(\salt, m, (h^{(e)}_0)_{e \in \oneto{\tau}})$$
where $\hash_1$ is an hash function, $h^{(e)}_0$ is the first-step commitment of the $e^{\text{th}}$ execution of the protocol, and $\salt$ a random value in $\{0, 1\}^{2\lambda}$. The prover obtains the first challenge from $h_1$ by using a XOF (extendable-output function).

The second challenge is generated in a similar way: the prover computes an element $h_2$ thanks to an other hash function $\hash_2$ and the other information computed during the step 3. The prover obtains the second challenge from $h_2$ by using a XOF.

The signature $\sigma$ therefore consists of sending:\begin{itemize}
    \item $\salt$ which allows to compute commitments and the Fiat-Shamir hashes,
    \item $h_1$ as commitment of the initial values,
    \item $h_2$ as hash of all responses of the first challenge,
    \item each response $\rsp$ of the second challenge.
\end{itemize}
The signature is:
$$\sigma = (\salt,h_1, h_2, (\rsp^{(e)})_{e \in \oneto{\tau}})$$

\subsection{MPC for proofs of knowledge}

\subsubsection{Multi-Party Computation}

We recall here the formalism proposed by \cite{FR22}. The sharing of a value $x$ into $N$ parties is denoted $(\share{x}_1,...,\share{x}_N)$, where $\share{x}_i$ is the share of index $i$ for $i\in\oneto{N}$, and $\share{x}_J=(\share{x}_i)_{i\in J}$ is the subset of shares for $J\subset\oneto{N}$.

\begin{definition}[Threshold LSSS]
    Let $\mathbb{F}$ a finite field and $t$ an integer in $\oneto{N}$. A $(t,N)$-threshold Linear Secret Sharing Scheme is a method to share a secret $s\in\mathbb{F}$ into $N$ shares $\share{s}=(\share{s}_1,...,\share{s}_N)\in\mathbb{F}^N$ such that the secret can be rebuilt from any subset of $t$ shares, while the knowledge of any subset of $t-1$ shares (or less) gives no information on $s$.
\end{definition}

Formally, a $(t,N)-$threshold LSSS is a pair of algorithms: $$\begin{cases}
\Share: \mathbb{F}\times R \rightarrow \mathbb{F}^N\\
\reconstruct_J: \mathbb{F}^t \rightarrow \mathbb{F}
\end{cases}$$

where $R$ denotes some randomness space, and $\reconstruct_J$ is indexed by a subset $J\subset\oneto{N}$ of $t$ elements. These algorithms must verify the following properties: \begin{itemize}
    \item \textbf{Correctness:} for every $x\in\mathbb{F}$, $r\in R$ and $J\in\oneto{N}$ such that $|J|=t$: let $\share{x}=\Share(x,r)$. We have: $$\reconstruct_J(\share{x}_J)=x$$
    \item \textbf{Perfect $(t-1)$-privacy:} For every $s_0,s_1 \in \mathbb{F}, I\subset \oneto{N}$ with $|I|=t-1$, the two distributions: \begin{align*}\Big\{ \share{s_0}_I \text{ }\vert\text{ } \genfrac{}{}{0pt}{0}{r \sampler R}{\text{ } \share{s_0}_{\oneto{N}} \longleftarrow \mathsf{Share}(s_0,r)}\Big\} \text{ and } \Big\{ \share{s_1}_I \text{ }\vert \text{ }\genfrac{}{}{0pt}{0}{r \sampler R}{\text{ } \share{s_1}_{\oneto{N}} \longleftarrow \mathsf{Share}(s_1,r)}\Big\}\end{align*} are perfectly indistinguishable.
\item \textbf{Linearity:} for every $v,v'\in\mathbb{F}^t$, $\alpha\in\mathbb{F}$ and subset $J\subset\oneto{N}$ of $t$ elements: $$\reconstruct_J(\alpha v+v')=\alpha\reconstruct_J(v)+\reconstruct_J(v')$$
\end{itemize}

Below are the two specific secret sharing which are used in RYDE:

\begin{definition}[Additive secret sharing]
    An additive secret sharing over a finite field $\mathbb{F}$ is an $(N,N)$-threshold LSSS where $R=\mathbb{F}^{N-1}$. The $\Share$ algorithm is defined as $$\Share: (s,(r_1,...,r_{N-1}))\rightarrow\share{s}=\left( r_1,...,r_{N-1},s-\sum_{i=1}^{N-1}r_i\right)$$
    The $\reconstruct_{\oneto{N}}$ algorithm consists in computing the sum of all the shares.
\end{definition}

\begin{definition}[Shamir's secret sharing]
    The Shamir's secret sharing over a finite field $\mathbb{F}$ is an $(\ell+1,N)$-threshold LSSS with $R=\mathbb{F}^{\ell}$. Let $e_1,...,e_N$ be public distinct elements of $\mathbb{F}^*$. The $\Share$ algorithm is defined as follows:\begin{itemize}
        \item Build a polynomial: $P(X)=s+\sum_{i=1}^\ell r_iX^i$
        \item For each $i\in\oneto{N}$: compute $\share{s}_i=P(e_i)$
    \end{itemize}
    The $\reconstruct_{J}$ algorithm consists in recovering $P$ with polynomial interpolation from known evaluations $\share{s}_J=(P(e_i))_{i\in J}$. The reconstructed secret is then given by the constant term $s$ of $P$.
\end{definition}

A multi-party computation protocol is an interactive protocol involving several parties whose objective is to jointly compute a function $f$ on the share $\share{x}$ they received at the begin of the protocol, so that they each get a share of $\share{f(x)}$ at the end of the execution. In the restricted MPC model described in~\cite{FR22}, at each step the parties can perform one of the following actions:
\begin{itemize}
\item Receiving elements: it can be randomness sent by a random oracle, or the shares of an hint which depends on the witness $\omega$ and the previous elements sent.
\item Computing: since the sharing is linear, the parties can perform linear transformations on the shared values.
\item Broadcasting: it is sometimes necessary to open some values (which give no information about the witness $\omega$) to execute the protocol.
\end{itemize}

For application of the MPCitH paradigm, this MPC protocol is only required to be secure in the semi-honest setting: we suppose that all parties rightfully perform their computation. The view of a party is made up of all the information available to this party when executing the protocol: its input shares and the messages it receives from other parties.

\begin{proposition}[\cite{FR22}]
    Given a $(\ell+1,N)$-threshold LSSS, the following property holds; for each $v\in\mathbb{F}^{\ell+1}$ and each subset $J\subset\oneto{N}$ of $\ell+1$ elements, there exists a unique sharing $\share{x}\in\mathbb{F}^{N}$ such that $\share{x}_J=v$, and such that for all $\mathcal{J}\subset\oneto{N}$ of $\ell+1$ elements: $$\reconstruct_\mathcal{J}(\share{x}_\mathcal{J})=\reconstruct_J(v).$$
\end{proposition}

One deduces that there exists an algorithm $\expand_J$ which returns the unique sharing from a subset $J$ of the shares. For example, in the Shamir's secret sharing, $\expand$ builds the Lagrange polynomial from the known evaluations and outputs the image of each party's point.

\subsubsection{MPC-in-the-Head paradigm}

The MPC-in-the-Head paradigm, introduced by \cite{IKOS07}, is a framework allowing to convert a generic secure multi-party computation (MPC) protocol to a zero-knowledge proof. Let us assume we have a $\ell$-private MPC protocol in which $N$ parties $\mathcal{P}_1, \ldots, \mathcal{P}_N$ securely and correctly evaluate a function $f$ on a secret input $x$. The verifier asks the prover to reveal the views of $\ell$ parties (which include input shares and communications with other parties). In this construction, the zero-knowledge property comes from the $\ell$-privacy of the MPC protocol.

In what follows, we assume that all the manipulated MPC protocols fit the model decribed in~\cite{FR22}: they use $(\ell+1,N)$-threshold LSSS and only uses the operations decribed in the previous subsection.

We want to build a zero-knowledge proof of knowledge of a witness for a statement $x$ such that $(x,\omega)\in\mathcal{R}$. Assume that we have a MPC protocol which evaluates a function $f$ on the witness $\omega$ and has the security properties mentioned above. The protocol works as follows:
\begin{itemize}
\item each party receives as input a share $\share{\omega}_i$, where $\share{\omega}$ is a threshold LSSS sharing of $\omega$,
\item the function $f$ outputs ACCEPT if $(x,\omega)\in\mathcal{R}$, REJECT otherwise,
\item the view of $\ell$ parties gives no information about the witness $\omega$.
\end{itemize}

In our case, the parties take as input a linear sharing $\share{\omega}$ of the secret $\omega$ (one share per party) and they compute one or several rounds in which they perform three types of actions:
\begin{description}
    \item[Receiving randomness:] the parties receive a random value $\epsilon$ from a randomness oracle $\oracle_R$. When calling this oracle, all the parties get the same random value $\epsilon$.
    \item[Receiving hint:] the parties can receive a sharing $\share{\beta}$ (one share per party) from a hint oracle $\oracle_H$. The hint $\beta$ can depend on the witness $w$ and the previous random values sampled from $\oracle_R$.
    \item[Computing \& broadcasting:] the parties can locally compute $\share{\alpha} := \share{\varphi(v)}$ from a sharing $\share{v}$ where $\varphi$ is an $\mathbb{F}$-linear function, then broadcast all the shares $\share{\alpha}_1$, \ldots, $\share{\alpha}_N$ to publicly reconstruct $\alpha := \varphi(v)$. The function $\varphi$ can depend on the previous random values $\{\epsilon^i\}_i$ from $\oracle_R$ and on the previous broadcast values. One should note that in the case of additive sharing, a single party needs to add a constant.
\end{description}


To build the zero-knowledge proof, the prover begins by building a sharing $\share{\omega}$ of $\omega$. The prover then simulates all the parties of the protocol until the computation of $\share{f(\omega)}$ (which should correspond to a sharing of ACCEPT, since the prover is supposed to be honest), and sends a commitment of each party's view. The verifier then chooses $\ell$ parties and asks the prover to reveal their views. The verifier checks the computation of their commitments. Since only $\ell$ parties have been opened, the revealed views give no information about $\omega$.

A naive way for the prover to execute the protocol would be to simulate all parties, but each party computes LSSS sharing of all involved elements. It implies that for each $a$ in these elements, a sharing $\share{a}=(\share{a}_1,...,\share{a}_N)$ contains redundancy when $\ell < N-1$. One deduces that it is necessary to perform the calculations of the associated MPC protocol for only a subset $S$ of $\ell+1$ out of $N$ parties, and commit them is enough to ensure the soundness of the proof. 
When $\ell+1$ is small, emulating the MPC protocol is cheap and leads to good overall performances: it is the high-level idea in Threshold-MPCitH. This technique to save computation is not possible when we use the additive sharing (since $\ell=N-1$). However, in the additive setting, \cite{AGHHJY22} showed that the prover only need to emulate $1+\log_2 N$ parties, and not $N$.


\subsubsection{Soundness}

Consider a malicious prover who want to convince the verifier that he knows a witness $\omega$ for the instance $x$ (although it does not). Two cases can happen:

\begin{itemize}
    \item The challenge of the verifier are such that the protocol results in a false positive,
    \item The malicious prover cheats for some party $i\in\oneto{N}$ (and hence has inconsistent view w.r.t. the commitments).
\end{itemize}
Feneuil and Rivain describe in \cite{FR22} a strategy allowing a malicious prover to convince the verifier with probability: $$\dfrac{1}{\binom{N}{\ell}}+p_\eta\:\dfrac{\ell\:(N-\ell)}{\ell+1}$$ where $p_\eta$ is the false positive rate of the associated MPC protocol (\textit{i.e.} is the probability that the MPC protocol output ACCEPT even if $(x,\omega)\not\in\mathcal{R}$). In the additive case ($\ell=N-1$), it leads to a soundness error of
$$\dfrac{1}{N}+p_\eta\:\left(1-\dfrac{1}{N}\right).$$

\subsection{Technical lemmas}

We present here technical lemmas that are useful in following proofs. The proofs of these lemmas can be found in the original references.

\begin{lemma}[Splitting Lemma \cite{PS}]
Let $A \subset X\times Y$ such that $\prb[(x,y)\in A] \ge \epsilon$. \\
For any $\alpha<\epsilon$, define $$B = \Big\{  (x,y)\in X\times Y \mid \prb_{y'\in Y}[(x,y')\in A] \ge \epsilon-\alpha  \Big\} \text{ and } \Bar{B} = (X\times Y)\setminus B$$ then:
\begin{itemize}
    \item $\prb[B]\ge \alpha$
    \item $\forall (x,y) \in B, \prb_{y' \in Y}[(x,y') \in A] \ge \epsilon - \alpha$
    \item $\prb[B | A] \ge \dfrac{\alpha}{\epsilon}$
\end{itemize}
\end{lemma}

\subsection{Mathematical background}

\subsubsection{$q$-polynomials}

We begin by recalling here the main properties of finite fields. Let $q$ a prime number and $m\in\mathbb{N}$. Remember that $\Fqm$ is a linear space over $\mathbb{F}_q$. The $q$-polynomials, introduced in \cite{ore}, are very useful tools in finite vector spaces: in particular, they allow to characterize linear subspaces.

\begin{definition}[q-polynomial]
A q-polynomial of q-degree $r$ is a polynomial in $\Fqm[X]$ of the form: \\
$$P(X) = X^{q^r}+\sum_{i=0}^{r-1}X^{q^i}p_i \qquad \text{with } p_i \in \Fqm$$
\end{definition}

\begin{proposition}
Let $P$ a $q$-polynomial in $\Fqm[X]$, $\alpha,\beta\in\Fq$, $x,y\in\Fqm$. Then: $$P(\alpha x+\beta y)=\alpha P(x)+\beta P(y)$$
\end{proposition}

\begin{proof}
Comes directly from the linearity over $\Fq$ of the Frobenius endomorphism: $x\rightarrow x^q$.
\end{proof}

One can then define a linear subspace in $\Fqm$ with a $q$-polynomial.

\begin{proposition}\label{sevroot1}
The set of roots of a non zero $q$-polynomial of degree $r$ is a linear subspace of dimension lower than or equal to $r$.
\end{proposition}

\begin{proof}
Let $P$ a $q$-polynomial of degree $r$. One can see $P$ as a linear application from $\Fqm$ to $\Fq^m$. As an endomorphism kernel, the set of zeros forms a linear space.\\
Since it is a polynomial of degree $q^r$, $P$ has at most $q^r$ roots, which allows to obtain the majoration of the dimension.
\end{proof}

\begin{proposition}[\cite{ore}]\label{sevroot2}
Let $E$ a linear subspace of $\Fqm$ of dimension $r\leq m$. Then there exists a unique monic $q$-polynomial of $q$-degree $r$ such that all element in $E$ is a root of $P$.\\
$P$ is called the annihilator polynomial of $E$.
\end{proposition}


\subsubsection{Rank metric}

While the Hamming metric -- which counts the number of nonzero coordinates of a code word -- is relevant in the case of a prime field (especially $\mathbb{F}_2$), the rank metric relies on an extension field $\mathbb{F}_{q^m}$.

\begin{definition}
Let $\bm{x}=(x_1,...,x_n)\in\mathbb{F}_{q^m}^n$, and $\mathcal{B}=(b_1,...,b_m)$ an $\mathbb{F}_q$-basis of $\mathbb{F}_{q^m}$. Each coordinate $x_i$ can be associated with a vector $(x_{i,1},...,x_{i,m})\in\mathbb{F}_{q}^m$ such that: $$x_i=\sum_{j=1}^m x_{i,j}b_j$$
The matrix $\bm{M}(\bm{x})=(x_{i,j})_{(i,j)\in\oneto{n}\times\oneto{n}}$ is called matrix associated to the vector $\bm{x}$.\\
The rank weight is defined as: $W_R(\bm{x})=\rank(\bm{M}(\bm{x}))$.\\
The distance between two vectors is given by: $d(\bm{x},\bm{y})=W_R(\bm{x}-\bm{y})$.\\
The support of $\bm{x}$ is the linear subspace of $\mathbb{F}_{q^m}$ generated by its coordinates: $\supp (\bm{x})=\langle x_1,...,x_n\rangle$
\end{definition}

\textit{Remark:} The rank weight and the support of a vector is independent of the choice of the basis. On the other hand, the rank weight of $\bm{x}$ is obviously equal to the dimension of its support.\\

It is clear that the metric takes on its full meaning in the case of field extensions: if we were to place ourselves on a prime field, any non zero vector would have weight exactly 1. It also follows that the rank metric is less discriminating since increased by the Hamming weight.

The number of possible $r$-dimensional supports in $\Fqm$ is given by the gaussian binomial coefficient: $$\cbg
=\prod_{i=0}^{r-1}\dfrac{q^m-q^i}{q^r-q^i}$$
One often use the following approximation: $\cbg\approx q^{r(m-r)}$.

\subsection{Rank Syndrome Decoding problem}

We want to build a signature based on a standard difficult problem: the syndrome decoding. However, rather than using the Hamming metric, we use here the rank metric, which allows to build similar constructions, but with different properties. We begin by recalling the fundamental definitions inherent to coding theory.

\begin{definition}
A linear code $\mathcal{C}$ on a finite field $\mathbb{F}$ of dimension $k$ and length $n$ is a linear subspace of $\mathbb{F}^n$ of dimension $k$. The elements in $\mathcal{C}$ are called the code words. One says that $\mathcal{C}$ is a $[n,k]_\mathbb{F}$ code.
\end{definition}

Two types of matrices can define a code:

\begin{definition}
Let $\mathcal{C}$ a $[n,k]_\mathbb{F}$ code.
A matrix $\bm{G}\in\mathbb{F}^{k\times n}$ is a generator matrix of $\mathcal{C}$ if its rows form a basis of $\mathcal{C}$, or equivalently: $$\mathcal{C}=\lbrace \bm{m}\bm{G}, \bm{m}\in\mathbb{F}^k\rbrace$$
A matrix $\bm{H}\in\mathbb{F}^{(n-k)\times n}$ is a parity check matrix of $\mathcal{C}$ if its rows form a basis of $\mathcal{C}^\perp$, or equivalently: $$\mathcal{C}=\lbrace \bm{c}\in\mathbb{F}^n,\bm{H}\bm{c}^\intercal=\bm{0}^\intercal\rbrace$$
\end{definition}

\textit{Remark:} In order to lighten the notations, we often omit the transpose symbol for the vertical vectors

As in the case of the Hamming metric with random codes, syndrome decoding is a hard problem in the rank metric:

\begin{definition}
Let $\mathbb{F}_{q^m}$ be the finite field of size $q^m$. Let $(n,k,r)$ be positive integers such that $k\leq n$. Let $I \in \mathbb{F}_{q^m}^{(n-k)\times (n-k)}$ the identity matrix of dimension $n-k$ on $\mathbb{F}_{q^m}$. The Rank Syndrom Decoding (Rank-SD) problem with parameters $(q,m,n,k,r)$ is the following one:\\
Let $\bm{H}$, $\bm{x}$ and $\bm{y}$ such that:\begin{itemize}
    \item $\bm{H}'$ is uniformly sampled from $\mathbb{F}_{q^m}^{(n-k)\times k}$, and $\bm{H} = (\bm{I} \,||\, \bm{H}') \in \Fqm^{\nmktn}$
    \item $\bm{x}$ is uniformly sampled from $\lbrace \bm{x}\in\Fqmn, W_R (\bm{x})\leq r\rbrace$
    \item $\bm{y}=\bm{H}\bm{x}$
\end{itemize}
From $(\bm{H},\bm{y})$, find $\bm{x}$.
\end{definition}

\textit{Remark:} Choosing the parity-check matrix $\bm{H}$ in standard form does not decrease the hardness of the problem, because obtaining this form from a random matrix can be done in polynomial time. As we see below, constraining this form allows us to share only the first $k$ coordinates of $\bm{x}$.

As we see later, there are two main approaches to solve the Rank-SD problem. The first one relies on combinatorial arguments: it consists in guessing the support of the error. The second one relies on algebraic computations, which consist in putting the problem into polynomial equations, then solving them using a Gröbner basis algorithm.

\section{Description of the signature scheme}

Here we explain how to build our signature scheme from a MPC protocol checking the solution of a Rank-SD instance. The MPC-in-the-Head paradigm allows us to convert the protocol first into a zero-knowledge proof, then into a signature scheme using the Fiat-Shamir transform.

\subsection{MPC rank checking protocol}

From the Proposition \ref{sevroot2} in the previous section, one can deduce a polynomial characterization of the rank of a vector: for all $\bm{x}\in\Fqm$, $W_R(\bm{x})\leq r$ if and only if there exists a $q$-polynomial $L$ of $q$-degree $r$ such that for all $i\in\oneto{n}$: $L(x_i)=0$. The following MPC protocol is an optimization of the one proposed in \cite{F22}.

We now describe a MPC protocol to check the maximum rank of a vector $\bm{x}=(x_1,...,x_n)\in\Fqm$. Let $U=\supp (\bm{x})=\langle x_1,...,x_n\rangle$ the vector space generated by its coordinates, and $L_U$ the annihilator polynomial of $U$: $$L_U=\prod_{u\in U} (X-u)\in\Fqm[X]$$
$L_U$ is a $q$-polynomial, so it can be written as: $$L_U(X)=\sum_{i=0}^{r-1}\beta_iX^{q^i}+X^{q^r}$$
Since the support is defined up to multiplication by a constant, one can always assume that it contains 1 (see \cite{RQC}) without impacting the security of the scheme. So, we get that $1$ is a root of the linearized polynomial $L_U$, we can consequently deduce a relation between its coefficients, and send one fewer value. We can suppose that $\bm{\beta}=(\beta_1,...,\beta_{r-1})\in\Fqm^{r-1}$ and $\beta_0= -\sum_{i=1}^{r-1}\beta_i-1$.

Rather than checking separately that each $x_i$ is a root of $L_U$, we batch all these checks by uniformly sampling $\gamma_1,...,\gamma_n$ in an extension $\mathbb{F}_{q^{m\cdot \eta}}$ of $\Fqm$ and checking that: $$\sum_{j=1}^n\gamma_jL_U(x_j)=0$$

If one of the $x_i$ is not a root of $L_U$, the equation above is satisfied only with probability $\frac{1}{q^{m\eta}}$. We have
\begin{align*}
\sum_{j=1}^n\gamma_jL(x_j) &= \sum_{j=1}^n\gamma_j \left(\sum_{i=0}^{r-1}\beta_ix_j^{q^i}+x_j^{q^r}\right)\\
&= \sum_{j=1}^n\gamma_jx_j^{q^r} +\sum_{i=0}^{r-1}\beta_i \cdot \sum_{j=1}^n\gamma_jx_j^{q^i}\\
&= \sum_{j=1}^n\gamma_jx_j^{q^r} + \beta_0\sum_{j=1}^n\gamma_jx_j + \sum_{i=1}^{r-1}\beta_i\cdot\sum_{j=1}^n\gamma_jx_j^{q^i} \\
&= \sum_{j=1}^n\gamma_jx_j^{q^r} + \Big(-\sum_{i=1}^{r-1}\beta_i-1\Big)\sum_{j=1}^n\gamma_jx_j + \sum_{i=1}^{r-1}\beta_i\cdot\sum_{j=1}^n\gamma_jx_j^{q^i} \\
&= \sum_{j=1}^n\gamma_jx_j^{q^r} + \Big(-\sum_{i=1}^{r-1}\beta_i\Big)\sum_{j=1}^n\gamma_jx_j + \sum_{i=1}^{r-1}\beta_i\cdot\sum_{j=1}^n\gamma_jx_j^{q^i} - \sum_{j=1}^n\gamma_jx_j \\
&= \sum_{j=1}^n\gamma_j(x_j^{q^r}-x_j) + \sum_{i=1}^{r-1}\beta_i\cdot\sum_{j=1}^n\gamma_j(x_j^{q^i}-x_j)
\end{align*}

Defining $z = -\sum_{j=1}^n\gamma_j(x_j^{q^r}-x_j)$, and $\omega_i= \sum_{j=1}^n\gamma_j(x_j^{q^i}-x_j) $, proving the equation is equivalent to proving that: $$z=\langle\beta,\omega\rangle$$ 
with $\beta$ the vector of coefficients of $L_U$. This can be checked in the same way as in the multiplication checking protocol from \cite{BN20}. Therefore, it is necessary to introduce a random vector $\bm{a}\in\Fqme^r$ and $c=-\langle\bm{\bm{\beta}},\bm{a}\rangle\in\Fqme$.

Since we consider a Rank-SD instance with $\bm{H}$ in standard form $\bm{H} = (\bm{I} \,||\, \bm{H}')$, the secret solution $\bm{x}\in\Fqm^n$ can be split it in two parts $\bm{x}=(\bm{x}_A\parallel \bm{x}_B)$ with $\bm{x}_A\in\Fqm^k$ and $\bm{x}_B\in\Fqm^{n-k}$, such that 
$$
\bm{y}=\bm{H}\bm{x} ~~~\Leftrightarrow~~~ \bm{x}_A = \bm{y}-\bm{H}'\bm{x}_B~.
$$
For a given instance, $(\bm{H},\bm{y})$, the secret solution can then be fully (and linearly) recovered from $\bm{x}_B$ by $\bm{x}= (\bm{y}-\bm{H}'\bm{x}_B\parallel\bm{x}_B)$.

The resulting MPC protocol, denoted $\Pi^\eta$, is presented in Figure \ref{mpc_opti}. In a high-level point of view, the difference between this MPC protocol and the one described in~\cite{F22} comes from the fact that we assume that $1$ is in the support of $x$ to save communication. However, this difference does not impact the security analysis (\textit{i.e.} the false-positive rate) of the protocol, see below.

\input{fig-mpc_rsd_opt}

\begin{proposition}[\cite{F22}] If $ \:W_R(\bm{x})\leq r$, then the protocol $\Pi^\eta$ always accepts. If $W_R(\bm{x})> r$, the protocol accepts with  probability at most: $p_\eta=\frac{2}{q^{m\eta}}-\frac{1}{q^{2m\eta}}$.\\
We call $p_\eta$ the false positive rate of $\Pi^\eta$.
\end{proposition}

\begin{proof}
The protocol takes place in two stages: from steps 1 to 4, it computes the values $z$ and $\omega$. From step 5, it checks that $z=\langle\bm{\beta},\bm{\omega}\rangle$. If $W_R(\bm{x})>r$, there exists $k\in \oneto{n}$ such that $L(x_k)\neq 0$.\\
There are two possibilities for the protocol to accept in this case:
\begin{itemize}
    \item either $\sum_{j=1}^n\gamma_jL(x_j)=0$, which occurs with probability $\frac{1}{q^{m\eta}}$;
    \item or $z\neq\langle \bm{\beta},\bm{\omega}\rangle$ but the protocol still accepts, which occurs with probability $\frac{1}{q^{m\eta}}$ (according to the false-positive rate of the \cite{BN20} multiplication checking protocol adapted for the matrix setting, see \cite{F22} for details).
\end{itemize}
Since these two events are independent, the false positive rate is: $$p_{\eta}=\frac{1}{q^{m\eta}}+\left(1-\frac{1}{q^{m\eta}}\right)\frac{1}{q^{m\eta}}\:\leq\:\frac{2}{q^{m\eta}}.$$
\end{proof}

\subsection{Conversion to zero-knowledge proof optimized with the hypercube technique} \label{mpc_to_zk}

We start by presenting the main idea of a proof of knowledge optimized with the hypercube technique as proposed by \cite{AGHHJY22}: sharing a secret in $N=2^D$ additive shares, and arranging them in a geometrical structure. Let us consider an hypercube of dimension $D$, each dimension having $2$ slots. These $2^D$ shares are called ``leaf parties''. They can be indexed in two ways: either by an integer in $\oneto{2^D}$, or by a vector of coordinates $(i_1,...,i_D)\in\oneto{2}^D$.

Since we use an additive secret sharing, one can intertwine several executions where each party (what we call a ``main party'') can be the sum of a subset of leaves, such that these subsets form a partition of the leaves. We build here $D$ intertwined protocols, and each of them has $2$ main parties. For each dimension $d\in\oneto{D}$, each main share is the sum of the $2^{D-1}$ leaves which have the same value for $i_d$. The main parties are indexed by $(d,k)\in\oneto{D}\times\oneto{2}$.

The leaf parties are the first to be generated. The process is the same as in the generation of additive shares in a standard MPCitH transformation: the prover uses a GGM tree (as a puncturable pseudo-random function) to derive $N$ seeds from a root seed $\seed$. From the $i$\textsuperscript{th} leaf seed, each party $i\in\oneto{N-1}$ pseudo-randomly generates the input shares $\share{\bm{x}_B}_i,\share{\bm{\beta}}_i,\share{\bm{a}}_i$ and $\share{c}_i$ (see Figure~\ref{mpc_opti} for the definition of these shares) . The last share for $i=N$ is computed such that we obtain a valid secret sharing of the values (except for $\bm{a}$ which is random and therefore can be sample from $\seed_N$). Then the prover can derive the shares of the ``main parties'' by summing the shares according to the partitions. For example, the share of $c$ of the main party $(1,k)$ is $\sum_{i_2,...,i_D}\share{c}_{(k,i_2,...,i_D)}$, which is the sum of the $2^{D-1}$ leaf shares such that $i_1=k$.

The resulting protocols are presented in Figure \ref{pzk_hypercube}, Figure \ref{exec_pi} and Figure \ref{check_pi}.

\input{fig-pzk_rsd_hyp}
\input{fig-exec_pi}
\input{fig-check_pi}

\bigskip

An important remark is that one could have chosen a number of main parties per dimension $n\neq 2$. The size of the signature depends on $N=n^D$, but not independently on $n$ and $D$. For example, choosing the MPC parameters $(n,D)=(2,8)$ or $(n,D)=(4,4)$ have no impact on the signature size, since $n^D=256$ in the both cases. However, it is in our interest to choose $n$ as small as possible: this increase the number $D$ of interleaved MPC protocols, but reduce the number of computations (which are much more expensive). Consequently, we always choose $n=2$ and select different dimensions $D$.

On the other hand, it is not necessary to do computations on all the main shares. Once one has done the calculations on the main parties on a dimension $d\in\oneto{D}$, one knows the value of $\bm{\alpha}$ for the protocol. Therefore, one can do the computations for only $n-1$ parties of protocols associated to all other dimensions, since one can deduce the value of the last share. This makes a total of only $n+(n-1)(D-1)$ computed parties, and choosing $n=2$ makes this value the smallest as possible with $N=n^D$ fixed. One obtains that it is necessary to do computations for only $D+1$ main parties.

\newpage

\subsubsection{Application of the Fiat-Shamir transform}

We deduce from the Fiat-Shamir transform a way to convert a proof of knowledge to a signature scheme, by transforming the proof in a non-interactive protocol. The signature scheme and its verification algorithm are depicted in Figure \ref{sign_hyp} and Figure \ref{ver_hyp}. Note that we use a value $\salt$, as in \cite{FJR22}, in order to increase the security of the scheme as it reduces the probability of seeds collision.

\input{fig-sign_hyp}

\input{fig-verify_hyp}

\textit{Remark:} At the step 11 of the signature scheme, sending the state of all shares $j\neq i^*$ consists on:\begin{itemize}
    \item If $i^*=N$, sending the sibling path associated to the share $N$;
    \item If $i^*\neq N$, sending the sibling path associated to the share $i^*$ and the shares $\share{\bm{x}_B^{(e)}}_{N}$, $\share{\bm{\beta}^{(e)}}_{N}$ and $\share{c^{(e)}}_N$ to recover all the states $j\neq i^*$.
\end{itemize}

We show in Section \ref{sec_sgn} that this scheme guarantees good security against EUF-CMA attacks.

\subsection{Conversion to zero-knowledge proof using low-threshold linear secret sharing}

\subsubsection{Application of the MPCitH paradigm}

We explain below that the MPC-in-the-Head paradigm is a framework allowing to convert an MPC protocol with threshold linear secret sharing to a zero-knowledge proof. Like in the subsection \ref{mpc_to_zk}, the explanations given below are correct only within the framework of the restricted model \cite{FR22}. The principle of the process is that the prover emulates in his head the parties, and the verifier asks to the prover to reveal the views of $\ell$ parties out of $N$ (which includes shares and communications with other parties). It ensures the zero-knowledge property as long as the underlying MPC protocol relies on a $(\ell+1,N)$-threshold secret sharing.

To build such a proof, the prover begins by computing a sharing $\share{\omega}$ of the witness $\omega$, solution of the Rank Syndrome-Decoding instance. The prover then simulates all the parties of the protocol until the computation of $\share{f(\omega)}$ (which should correspond to a sharing of ACCEPT, since the prover is supposed to be honest), and sends a commitment of each party's view to show good faith. The verifier then chooses $\ell$ parties and asks to the prover to reveal their views. The verifier checks the computation of the their commitments. Since only $\ell$ parties have been opened, the revealed views give no information about $\omega$.

~
\newpage
~

A naive way for the prover to execute the protocol would be to simulate all parties, but each party computes LSSS sharing of all involved elements. It implies that for each $x$ in these elements, a sharing $\share{x}=(\share{x}_1,...,\share{x}_N)$ contains redundancy. One deduces that it is necessary to perform the calculations of the associated MPC protocol for only a subset $S$ of $\ell+1$ out of $N$ parties, and commit them is enough to ensure the soundness of the proof. The threshold protocol therefore ensures better performances for the simulation of the MPC protocol than the additive case (especially when $\ell$ is small compared to $n$), since it is not necessary to compute all the parties.

\textit{Remark:} This advantage on the full execution of the signing algorithm is partially mitigated by the performances of the symmetric cryptography primitives involved in the protocol. Threshold-MPCitH has better performances on the pseudo-random generation, but the commitment is faster (for the prover) with additive sharing. See \cite{FR22} for more details.

The protocol resulting from these considerations can be found in Figure \ref{pzk_thr}.

\input{fig-pzk_rsd_thr.tex}

\subsubsection{Digital signature scheme from Rank-SD problem}

In this section, one applies the Fiat-Shamir transformation to the previous interactive zero-knowledge proof in order to obtain a signature scheme. One transforms the zero-knowledge proof into a non-interactive protocol by deterministically sampling challenges thanks to a hash function. It must be run several times to achieve the desired level of security.

The resulting signature scheme is depicted in the Figure \ref{sign_thr}. 

\input{fig-sign_thr}

An entity that wants to verify the signature must then recompute the challenges using the hashes $h_1$ and $h_2$ provided in the signature. For each iteration $e\in\oneto{\tau}$, the prover sends $\rsp^{(e)}$, which includes all the shares belonging to the parties in $I$, and $\share{\boldsymbol{\alpha}^{(e)}}_{i^{*(e)}}$ which allows to recompute $\boldsymbol{\alpha}^{(e)}$. As a result, this information is sufficient for the verifier to check the hashes. Since hash functions are assumed to be collision free, checking equality between the hashes is sufficient to assume the validity of the signature.

The verification algorithm is formally detailed in the Figure \ref{ver_thr}.

\input{fig-verify_thr}

As for the additive sharing case, we prove in the Section \ref{sec_sgn} that this scheme guarantees good security against EUF-CMA attack.

\subsubsection{Using Shamir secret sharing}
\label{sec:using-sss}

In practice, we use the $(\ell+1)$-threshold Shamir's secret sharing scheme as LSSS. Thus the shares of a secret value $s\in\mathbb{F}_q$ are defined as $\share{s}_i=P(e_i)$, where $P$ is a polynomial whose constant term is equal to $s$, and $e_1,...,e_N\in\mathbb{F}_q$ are public non zero distinct points.
However, since each share corresponds to the evaluation of a polynomial into a distinct point of $\mathbb{F}_q$, we get that the number $N$ of parties is upper bounded by $q$. To have interesting performances, it means that $q$ must be large (for example, $q=256$).

It is possible to mitigate this constraint. Let us assume that we want to share values directly from $\Fqm$ (instead of $\mathbb{F}_q$). In that case, $e_1,\ldots,e_N$ would belong to $\Fqm$ and we would have $N \leq q^m$ instead of $N \leq q$. The only issue comes from the usage of the Frobenius endormorphism which is $\mathbb{F}_q$-linear and not $\Fqm$-linear. We can remark that, from a sharing $\share{s}$ of $s$, each party $i$ can obtain a sharing of $s^q$ by simply computing $\share{s}_i^q$. However, the evaluation points of the obtained sharing and the encoding polynomial changed: they become $e_1^q,...,e_N^q$. Indeed,
\begin{align*}
    P(X)^q &= \left(s+\sum_{i=1}^\ell r_iX^i\right)^q\\
    &= s^q+\sum_{i=1}^\ell r_i^q(X^q)^i\\
    &= \Tilde{P}(X^q)\qquad\text{where}\quad \Tilde{P}=s^q+\sum_{i=1}^\ell r_i^qX^i
\end{align*}
One can add two shares if and only if they have the same evaluation points. The optimized MPC protocol described used in Figure~\ref{sign_thr} does not satisfy this property: we can not set the input sharings such that we have the guarantee that addition is performed only on sharings with the same evaluation points. However, if we want to have a scheme relying on Threshold-MPCitH with a small $q$ (for example, with $q=2$), we can use the MPC protocol proposed in \cite{F22}. The latter satisfies the desired properties, see~\cite[Appendix A]{F22} for details.


\section{Parameter sets}

\subsection{Choice of parameters}

The problem Rank Syndrome Decoding is parameterized by the following parameters: \begin{itemize}
    \item $q\in\mathbb{N}$ - the order of the base field on which the problem is based
    \item $m\in\mathbb{N}$ - the degree of the considered extension of the above field 
    \item $n\in\mathbb{N}$ - the length of the code $\mathcal{C}$
    \item $k\in\mathbb{N}$ - the dimension of the code
    \item $r\in\mathbb{N}$ - the rank of the vector $\bm{x}$
\end{itemize}

The other parameters of the scheme are:\begin{itemize}
    \item $N\in\mathbb{N}$ - the number of parties simulated in the MPC protocol
    \item $\tau\in\mathbb{N}$ - the  number of rounds in the signature
    \item $\eta\in\mathbb{N}$ - allowing to build $\Fqme$
\end{itemize}

In order to choose the parameters, we need to consider:
\begin{itemize}
    \item The security of the Rank-SD instance, i.e, the complexity of the attacks on the chosen parameters.
    \item The security of the signature scheme, i.e, the cost of a forgery.
    \item The size of the signature
\end{itemize}

Since the protocol is more efficient as $N$ grows (up to some point), we need to take a larger value of $q$ in the threshold variant. To have a threshold scheme with $N=256$ shares, we must take $q=256$ since we necessarily have $N \le q$ (in the Shamir Secret Sharing scheme, we can't have more shares than the number of elements in the base field; this means we can't have more parties than the cardinal of $\Fq$ in our case).

In additive sharing, we take the value of $q=2$, since it is the most efficient and we have no constraints. We propose here two versions of the signature scheme with different values for $N$: we set $N = 256$ for a short version of the signature, and $N = 32$ for a fast one.

The choice of $(q,m,n,k,r)$ was then made in order to achieve the target security in regards to the attacks we present in Subsection \ref{atk_rsd}.

As for the choice of $\tau$ and $\eta$, we need to choose them such that our signature scheme resists to the forgery attack described by \cite{KZ20}, which we explain in Subsection \ref{atk_fs}. Because of this, in the additive protocol optimized with the hypercube technique, we need to set $\tau$ such that:
\begin{equation} \text{cost}_{\text{forge}}=\min_{0\leq\tau'\leq\tau}\left\lbrace\dfrac{1}{\sum_{i=\tau'}^\tau \binom{\tau}{i}p^i(1-p)^{\tau-i}}+N^{\tau-\tau'}\right\rbrace
\label{kz_eq_hyp}\end{equation} is higher than $2^\lambda$, with $p = \frac{2}{q^{m\eta}}-\frac{1}{q^{2m\eta}}$. The value depends on $\tau$ and $\eta$, we then need to take them such that the signature is safe, and as small as possible, given the parameters of the Rank-SD instance.

In the case of the threshold protocol, this formula changes:\begin{equation}
\text{cost}_{\text{forge}}=\min_{0\leq\tau'\leq\tau}\left\lbrace\dfrac{1}{\sum_{i=\tau'}^\tau \binom{\tau}{i}p^i(1-p)^{\tau-i}}+\binom{N}{\ell}^{\tau-\tau'}\right\rbrace \label{kz_eq_th}\end{equation}  with $p=( \frac{2}{q^{m\eta}}-\frac{1}{q^{2m\eta}})\cdot \binom{N}{\ell+1}$, see~\cite{FR22} for details. \\
In any case, as long as the cost of the forgery is high enough, we can just take the parameters for which the size of signature is the smallest.

We propose the following parameters for a security of $\lambda=128$ bits: \begin{itemize}
    \item For the additive scheme: $(q,m,n,k,r)=(2,31,33,15,10)$
    \item For the threshold scheme with $\ell=3$: $(q,m,n,k,r)=(256,11,12,5,5)$
    \item For the threshold scheme with $\ell=3$ with small $q$: $(q,m,n,k,r)=(2,31,29,14,10)$
\end{itemize}

\subsection{Signature and key sizes}

\subsubsection{Additive MPC}

We must send for each $e\in\oneto{\tau}$ the following elements: $\cmt_{i^*}\in\{0,1\}^{2\lambda}$ and $(\state^{(e)}_i)_{i\neq i^*}$ (which consists in seeds in $\{0,1\}^{2\lambda}$ or the following elements for one of them: $\share{\bm{x}_B}\in\Fqm^k$, $\share{\bm{\alpha}}\in\Fqme^{r-1}$, $\share{\bm{\beta}}\in\Fqm^{r-1}$, $\share{c}\in\Fqme$); and $(\salt\vert h_1\vert h_2)\in\{ 0,1\}^{6\lambda}$.

Note also that for each $e\in\oneto{\tau}$, we do not send all the $N-1$ seeds $(\state^{(e)}_i)$, since we use a tree-PRG (noted TPRG in the figures) to generate them: starting from the root seed, each seed is recursively expanded in two child seeds until we obtain $N$ seeds. We obtain a tree of depth $\log_2 N$, and we only reveal the $\log_2 N$ sibling nodes of the root-leaf path of the hidden leaf $i^*$ to reveal all the leaves except $i^*$.

We can now compute the communication cost of the protocol (in bits):
\begin{equation*}
\text{$|\sigma|$} = \underbrace{6\lambda}_\text{$\salt,h1,h2$}+\: \tau \cdot \left( \Big(\underbrace{(r-1) m \eta}_\text{$\bm{\alpha}$} +\underbrace{k m}_\text{$\bm{x}_B$}+\underbrace{(r-1)m}_\text{$\bm{\beta}$}+\underbrace{m \eta}_\text{$c$}\Big)\cdot \log_2 q + \underbrace{2\lambda + (\log_2N)\lambda}_\text{additive MPCitH}\right)\\
\end{equation*}

Parameters are chosen such the Rank-SD problem is secure against the existing attacks (see Section \ref{atk_rsd}). We recall that the security levels of the signature I, III and V correspond respectively to the security of \texttt{AES-128}, \texttt{AES-192} and \texttt{AES-256}. Tables \ref{table_add_short} and \ref{table_add_fast} give the theoretical signature size for different set of secure parameters in the case of additive shares: \input{param-hyp}

\subsubsection{Threshold MPC}

In the threshold case, the process is similar to the additive one above: we must send for each $e\in\oneto{\tau}$ the following elements: $h\in\{0,1\}^{2\lambda}$ (an hash used) and $\ell$ times all the elements$(\state^{(e)}_i)_{i\neq i^*}$; and one time $(\salt\vert h_1\vert h_2)\in\{ 0,1\}^{6\lambda}$.


The signature size of the scheme is given by
\begin{equation*}
    |\sigma| \leq \underbrace{6 \lambda}_{\salt,h1,h2} +\: \tau \cdot \Big( \ell \cdot \big(\underbrace{km}_{\bm{x}_B} + \underbrace{(r-1)m\eta}_\text{${\bm{\alpha}}$} + \underbrace{(r-1)m}_\text{$\bm{\beta}$} + \underbrace{r m\eta}_\text{$\bm{a}$ and $c$}\big)\cdot \log_2(q) + \underbrace{2\lambda\cdot \ell \log_2(\frac{N}{\ell})}_\text{Threshold MPCitH} \Big)
\end{equation*}
Let us recall that the scheme requires that $N \leq q$. If we want to use a small $q$, we need to replace our MPC protocol with the one described in \cite{F22} (see Section~\ref{sec:using-sss} for details). In that case, the signature size is given by
\begin{equation*}
    |\sigma| \leq \underbrace{6 \lambda}_{\salt,h1,h2} +\: \tau \cdot \Big( \ell \cdot \big(\underbrace{km}_{\bm{x}_B} + \underbrace{rm\eta}_\text{${\bm{\alpha}}$} + \underbrace{(r-1)m}_\text{$\bm{\beta}$} + \underbrace{r m\eta}_\text{$\bm{a}$ and $c$}\big)\cdot \log_2(q) + \underbrace{2\lambda\cdot \ell \log_2(\frac{N}{\ell})}_\text{Threshold MPCitH} \Big)
\end{equation*}

Tables \ref{table_thr_l3} and \ref{table_thr_l1} give the theoretical signature size for different set of secure parameters in the case of threshold shares for $\ell=3$ and $\ell=1$:
\input{param-l3}

\section{Security analysis}

\subsection{Security proof for the proof of knowledge}

\subsubsection{Additive MPC}

We prove here that our protocol fulfils the three required properties: completeness, soundness and zero-knowledge. The proofs of this section are close to those in \cite{AGHHJY22} (which are themselves close to \cite{FJR22}) due to the same structure of hypercube used, as only the MPC protocol changed.

We firstly prove that an honest prover is always accepted. Conversely, a prover who commits a bad witness in the first step of the protocol has a probability lower than $\epsilon=\left(p+(1-p)\frac{1}{N}\right)$ of being accepted. Consequently, a prover that has a higher probability of acceptance knows the secret.

\begin{theorem}
The protocol in Figure \ref{pzk_hypercube} is perfectly complete. Namely, a prover that knows $x$ and executes the protocol correctly is accepted by a verifier with probability 1.
\end{theorem}

\begin{proof}
    By construction, if the prover has knowledge of a solution of the Rank-SD instance, he is always able to execute the protocol correctly.
\end{proof}

\begin{theorem} \label{sound_hyp}
    If an efficient prover $\pt$ with knowledge of only public data $(\bm{H},\bm{y})$ can convince a verifier $\ver$ with probability: $$\Tilde{\epsilon}=\prb\big(\langle\pt,\ver\rangle\rightarrow ACCEPT\big)>\epsilon = p+(1-p)\dfrac{1}{N}$$
    then there exists an extraction function $\mathcal{E}$ that produces a commitment collision, or a good witness $\bm{x'}$ solution of the Rank-SD instance, by making an average number of calls to $\pt$ upper bounded by: $$\dfrac{4}{\Tilde{\epsilon}-\epsilon}\left(1+\dfrac{2\Tilde{\epsilon}\ln 2}{\Tilde{\epsilon}-\epsilon}\right)$$
\end{theorem}

\begin{proof} The proof is similar to that of \cite{AGHHJY22}.
We first need to establish that the probability for the malicious prover (who has no knowledge of the solution of the Rank-SD instance used) to cheat is at most $\epsilon = \frac{1}{N}+(1-\frac{1}{N})\cdot p$. \\
There are two situations where a malicious prover can be accepted by the verifier: \begin{itemize}
    \item He obtains the value $v=0$ when executing the MPC protocol thanks to a false positive.
    \item The verifier believes that the value obtained $v$ is $0$.
\end{itemize}
The first case occurs with probability $p = \frac{2}{q^{m\eta}}-\frac{1}{q^{2m\eta}}$, as false positive rate of the protocol $\Pi^\eta$.

In the second case, the malicious prover needs to alter the communications in order to fool the verifier. Among the $N$ shares, only the share $i^*$ is not be revealed by the prover. This means that, if he cheats on more than one share, the verifier notices the cheating, and thus rejects the proof. However, if he cheats on zero share, he is rejected as well since the value $v$ isn't $0$ since the malicious prover does not have a good witness.

This means the malicious prover can only cheat on one share exactly. However, cheating on one share $s$ implies that this it appears on one main share of all the $D$ dimensions, as the main shares are the sum of leaves that have the same index $i_k$ along the current dimension. This means that the cheating is not detected if and only if the share the prover cheated on is $i^*$, since there exists a bijection between leaves and the set of their associated main parties. This happens with probability $\frac{1}{N}$ (as the prover does not know the value of $i^*$ before cheating). No other cheating is possible, because all leaf parties except $i^*$ are revealed. Remark that cheating on the main parties give the same probability of success, since the prover should cheat for each dimension on the main party associated to the unrevealed leaf $i^*$.

Thus, the probability for the malicious prover to fool the verifier is at most: $p+(1-p)\cdot \frac{1}{N} = \frac{1}{N}+(1-\frac{1}{N})\cdot p$.\\ [0.3cm]
We then need to show that the soundness is exactly this value, which we do in what follows: \\
let $\mathcal{T}_1$ and $\mathcal{T}_2$ two transcripts with the same commitments, i.e, the same $h_0 = \hash({\cmt_1, \dots, \cmt_{N}})$ , but the second challenges $i^*_1$ (for $\mathcal{T}_1$) and $i^*_2$ (for $\mathcal{T}_2$) differ.

Then, we have two possibilities: \begin{itemize}
    \item $\share{\bm{x}_B}, \share{\bm{\beta}}, \share{\bm{a}}$ and $\share{c}$ differ in the two transcripts, and the malicious prover found a collision in the commitment hash.
    \item the openings of the commitments are equal, and thus the shares $\share{\bm{x}_B}, \share{\bm{\beta}}, \share{\bm{a}}$ and $\share{c}$ are equal in the transcripts.
\end{itemize}
We only consider the second case, as we suppose that we use secure hash functions and secure commitment schemes.

Then, since $i^*_1$ and $i^*_2$ are different challenges and the commitments are the same, it is possible to recover the witness. We show further that this means we can build an extractor function in order to obtain a good witness.

$\bm{x}$ is called a good witness if it is solution to the Rank-SD instance defined by public data, i.e. $\bm{H}\bm{x}=\bm{y}$ and $W_R(\bm{x}) \le r$. Let $R_h$ the random variable associated to the randomness in initial commitment, and $r_h$ is the value it takes.

For a malicious prover $\pt$ to be able to get these two transcripts, he does the following: \begin{itemize}
    \item Run the protocol with randomness $r_h$ with the verifier until $\mathcal{T}_1$ is found, i.e, $\mathcal{T}_1$ is the first accepted transcript found by $\pt$. We note $i^*_1$ the leaf challenge obtained.
    \item Then, using the same randomness $r_h$ that was used, i.e, building the same commitments, $\Tilde{\mathcal{P}}$ repeats the process, until he finds another accepted transcript, $\mathcal{T}_2$, for which the leaf challenge, $i^*_2$, is different than $i^*_1$.
    \item He then recovers the witness. If it is a bad witness, he repeats from the beginning.
\end{itemize}

In order to establish the soundness of the protocol, we need to estimate the number of times a malicious prover needs to repeat the authentication protocol in order to get the good witness $\bm{x}$. \\
Let $\delta \in ]0,1[$, and $\Tilde{\epsilon}$ such that $(1-\delta)\cdot\Tilde{\epsilon} > \epsilon$. We define the randomness $r_h$ to be a good randomness if $\prb(\operatorname{Succ}_{\pt}\vert r_h ) > (1-\delta)\cdot\Tilde{\epsilon}$.

By the Splitting Lemma, we have that $\prb (r_h\: good | \operatorname{Succ}_{\Tilde{\mathcal{P}}}) \ge \delta$. This means that after $\frac{1}{\delta}$ accepted transcripts, we have good odds to obtain a good randomness. Furthermore, we know that if the malicious prover uses a bad witness, his probability to cheat is bounded from above by $\epsilon$. Since the probability of success is greater than $\epsilon$, this means that a good witness has been used (when $r_h$ is good). 

To continue this proof, we look at the probability to have, given an accepted transcript $\mathcal{T}_1$, a second accepted transcript, $\mathcal{T}_2$, with a challenge different than the one in $\mathcal{T}_1$. This means we are looking to bound from below the probability: 
$$ \prb(\operatorname{Succ}_{\Tilde{\mathcal{P}}} \cap (i^*_1 \ne i^*_2) | r_h \: good )$$
Trivially, we know that this probability is equal to the probability of success knowing that $r_h$ is good, minus the probability of success with $i^*_1 = i^*_2$ knowing $r_h$ is good. This means: \small\begin{align*}   
\prb(\operatorname{Succ}_{\Tilde{\mathcal{P}}} \cap (i^*_1 \ne i^*_2) | r_h \text{ }good ) &=\prb(\operatorname{Succ}_{\Tilde{\mathcal{P}}} | r_h \text{ }good ) - \prb(\operatorname{Succ}_{\Tilde{\mathcal{P}}} \cap (i^*_1 = i^*_2) | r_h \text{ }good )\\ & \ge \prb(\operatorname{Succ}_{\Tilde{\mathcal{P}}}|r_h \: good) - \frac{1}{N}\\ & \ge (1-\delta)\Tilde{\epsilon}-\frac{1}{N} \\ & \ge (1-\delta)\Tilde{\epsilon}-\epsilon\qquad \text{ (since $\epsilon \ge \frac{1}{N}$ trivially)}\end{align*}\normalsize

Now that we have this lower bound, we want to estimate the number of time one has to repeat the protocol to find $\mathcal{T}_2$. For that, we take the opposite probability, i.e, $1-\prb(\operatorname{Success}_{\Tilde{\mathcal{P}}} \cap (i^*_1 \ne i^*_2) | r_h \text{ }good )$, which is lower bounded by $1-((1-\delta)\Tilde{\epsilon}-\epsilon)$. We now want a probability of $\frac{1}{2}$ at least of success after $L$ tries of the authentication protocol. This means then that we want: \begin{align*}
\Big(1-\prb(\operatorname{Success}_{\Tilde{\mathcal{P}}} \cap (i^*_1 \ne i^*_2) | r_h \text{ }good )\Big)^L &< {\frac{1}{2}} \\
L \cdot \ln(1-((1-\delta)\Tilde{\epsilon}-\epsilon)) & < -\ln(2)\\
L &> -\frac{\ln(2)}{\ln(1-((1-\delta)\Tilde{\epsilon}-\epsilon))}\end{align*}
One obtains the following majoration for the number of calls to $\pt$:
$$L>\frac{\ln(2)}{\ln(\frac{1}{1-((1-\delta)\Tilde{\epsilon}-\epsilon)})} \approx \frac{\ln(2)}{(1-\delta)\Tilde{\epsilon}-\epsilon} $$
This means that, when repeating the protocol $L$ times, the probability to get the second transcript is higher than $\frac{1}{2}$.

Finally, we can look at the number of protocol repetitions that have to be done. 
To quickly remind the steps of the extraction: \begin{itemize}
    \item $\Tilde{\mathcal{P}}$ repeats the authentication protocol until he finds an accepted transcript $\mathcal{T}_1$, where the commitments are generated by $r_h$, and with second challenge $i^*_1$.
    \item When $\mathcal{T}_1$ is found, repeat the protocol with the same value $r_h$, $L$ times. After that, $\Tilde{\mathcal{P}}$ has more than $\frac{1}{2}$ chance of being successful. If he is not, he repeats from the first step of the procedure.
\end{itemize}

We note $\mathbb{E}(\Tilde{\mathcal{P}})$ the number of repetitions $\Tilde{\mathcal{P}}$ has to make. After $L$ calls (to find $\mathcal{T}_2$), if $r_h$ is good (which happens with probability $\delta$), we have $\frac{1}{2}$ chance of not finding $\mathcal{T}_2$. However, if $r_h$ is not good (with probability $1-\delta$), then we can consider that $\mathcal{T}_2$ is never found.

Thus, $\prb(\text{no }\mathcal{T}_2 | \operatorname{Succ}_{\Tilde{\mathcal{P}}}) = \frac{\delta}{2}+ (1-\delta) = 1-\frac{\delta}{2}$. If that happens, then, $\Tilde{\mathcal{P}}$ has to return to the first step, i.e, find $\mathcal{T}_1$ again. This means:$$ \mathbb{E}(\Tilde{\mathcal{P}}) \le 1+{\Big((1-\prb(\operatorname{Succ}_{\Tilde{\mathcal{P}}}))\mathbb{E}(\Tilde{\mathcal{P}}) \Big)} + \prb(\operatorname{Succ}_{\Tilde{\mathcal{P}}})\Big( L+(1-\frac{\delta}{2})\mathbb{E}(\Tilde{\mathcal{P}})\Big)$$

Obviously, $\Tilde{\mathcal{P}}$ needs to run at least once. Then, we need to add to that the number of times expected before finding $\mathcal{T}_1$, and then, the number of times expected before finding $\mathcal{T}_2$. \\
Since $\prb(\operatorname{Succ}_{\Tilde{\mathcal{P}}}) = \Tilde{\epsilon}$ (by assumption), we can replace, and simplify the expression. We find then: \begin{align*}
    \mathbb{E}(\Tilde{\mathcal{P}}) & \le 1+{\Big((1-\Tilde{\epsilon}\cdot\mathbb{E}(\Tilde{\mathcal{P}}) \Big)} + \Tilde{\epsilon}\cdot\Big( L+(1-\frac{\delta}{2})\cdot\mathbb{E}(\Tilde{\mathcal{P}})\Big) \\ \mathbb{E}(\Tilde{\mathcal{P}}) & \le 1+\mathbb{E}(\Tilde{\mathcal{P}})-\Tilde{\epsilon}\cdot \mathbb{E}(\Tilde{\mathcal{P}}) + \Tilde{\epsilon} \cdot L + \Tilde{\epsilon} \cdot\mathbb{E}(\Tilde{\mathcal{P}}) - \Tilde{\epsilon} \cdot \frac{\delta}{2}\cdot \mathbb{E}(\Tilde{\mathcal{P}})\\
    \Tilde{\epsilon}\cdot\frac{\delta}{2}\cdot \mathbb{E}(\Tilde{\mathcal{P}}) &\le 1+ \Tilde{\epsilon}\cdot L \\
    \mathbb{E}(\Tilde{\mathcal{P}}) &\le \frac{2}{\Tilde{\epsilon}\cdot \delta}\Big(1+\Tilde{\epsilon}\cdot L\Big) = \frac{2}{\Tilde{\epsilon}\cdot \delta}\Big(1+\Tilde{\epsilon}\cdot \frac{\ln(2)}{(1-\delta)\Tilde{\epsilon}-\epsilon}\Big)
\end{align*}

Since this equality holds for any $\delta \in ]0,1[$, we can take $\delta$ such that $(1-\delta)\Tilde{\epsilon} = \frac{1}{2}(\Tilde{\epsilon}+\epsilon)$, and thus, we obtain the result: \begin{align*}
    \mathbb{E}(\Tilde{\mathcal{P}}) \le \frac{4}{\Tilde{\epsilon}-\epsilon}\cdot \Big(1+2\Tilde{\epsilon}\cdot \frac{\ln(2)}{\Tilde{\epsilon}-\epsilon}\Big)
\end{align*}
This means we found an upper bound on the number of iterations $\Tilde{\mathcal{P}}$ has to do before being able to retrieve a good witness, in the case where the probability to cheat was higher than $\epsilon$. This proves that $\epsilon$ is exactly the value $\frac{1}{N}+(1-\frac{1}{N})\cdot p$.
\end{proof}

We now have to prove zero-knowledge property by building a simulator which outputs indistinguishable transcripts with the distribution of transcripts from honest executions of the protocol.

The main intuition is that if the prover knows the challenge before committing the initial shares, then he knows for which party to cheat. We consequently build an HVZK simulator which firstly generate the challenge to build valid transcriptions of the protocol.

\begin{theorem}
    If the algorithm PRG is a secure pseudo-random generator and the commitment $\mathsf{Com}$ is hiding, then the algorithm \ref{hvzk_hyp} is Honest-Verifier Zero Knowledge.
\end{theorem}

\begin{proof}
Consider a simulator, described in Figure \ref{hvzk_hyp}, which produces the transcript responses $(h_0,\ch_1,h_1,\ch_2,\rsp)$. We demonstrate that this simulator produces indistinguishable transcripts from the real distribution (the one that we would obtain if it were generated by an honest prover who knows $\bm{x}$) by considering a succession of simulators: we begin by a simulator which produces true transcripts, and and change it gradually until arriving the following simulator. We explains why the distribution of transcripts is always the same at each step.

\input{fig-hvzk_hyp}

\begin{itemize}
    \item \textbf{Simulator 0 (real world):} It correctly executes the algorithm \ref{pzk_hypercube}, hence its output is the correct distribution.
    
    \item \textbf{Simulator 1:} Same as the \textbf{Simulator 0}, but starts by sampling the random challenges, and then uses true randomness instead of seed-derived randomness for leaf $i^*$. If $i^*=N$, the leafs $\share{\bm{x}_B}_{N},\share{\bm{\beta}}_{N},\share{c}_{N}$ are computed as in the MPC protocol.
    
    The pseudo-random generator is supposed to be $(t,\epsilon_{PRG})$-secure, its outputs are indistinguishable from the uniform distribution. Since the principal parts correspond to the sum of a certain number of leaves whose distributions are indistinguishable from that of the real world, their distribution is also indistinguishable.

    \item \textbf{Simulator 2:} Replace the leafs $\share{\bm{x}_B}_{N},\share{\bm{\beta}}_{N},\share{c}_{N}$ in \textbf{Simulator 1} by uniform sampled values. Compute $\share{v}_{i^*}=-\sum_{i\neq i^*}\share{v}_i$. Note that this simulator becomes independent of the secret witness $\bm{x}_B$.

    If $i^*=N$, it only impacts the shares $\share{\bm{\alpha}}_{i^*}$ and $\share{v}_{i^*}$. Note that this change does not alter the uniform distribution of these values. It does not alter the distribution of any other leaf.

    If $i^*\neq N$, it only impacts $\aux_N=\left(\share{\bm{x}_B}_{N},\share{\bm{\beta}}_{N},\share{c}_{N}\right)$ in the simulated response and the values computed from them in the MPC protocol. It does not alters the distribution of other leaves. We observe that the shares in $\aux_N$ are calculated by adding a randomness value from each seed of party $i\neq i^*$, which amounts to adding a uniform random value from $\seed_{i^*}$. Since this distribution was uniform too in \textbf{Simulator 1}, the output distributions are the same. Remember that this does not change the distributions of the main parts for the same reasons as before.

    \item \textbf{Simulator 3:} Rather than computing the value of $\share{\bm{\alpha}}_{i^*}$ as in the MPC protocol, Simulator 3 samples it uniformly at random from $\Fqme^{r-1}$. As in the previous simulator, it does not change their output distribution.
\end{itemize}

As such, the output of the simulator is indistinguishable from the real distribution.
\end{proof}

\subsubsection{MPC with threshold}

We discuss here the security properties of the threshold protocol. The following proofs are only sketched, as they are very similar to those of the hypercube protocol.

\begin{theorem}
The protocol in Figure \ref{pzk_thr} is perfectly complete. Namely, a prover who knows $x$ and performs the protocol correctly is accepted by a verifier with probability 1.
\end{theorem}

\begin{proof}
    Same trivial justification as in the previous case of the hypercube.
\end{proof}

\begin{theorem}
    Suppose that there is an efficient prover $\pt$ that on input $(\bm{H},\bm{y})$ convinces an honest verifier $\mathcal{V}$ to accept with probability $$\Tilde{\epsilon}=\prb\big(\langle\pt,\ver\rangle\rightarrow ACCEPT\big)>\epsilon = \dfrac{1}{\binom{\ell}{N}}+p\cdot\dfrac{\ell\cdot (N-\ell)}{\ell+1}$$
    then there exists an extraction algorithm $\mathcal{E}$ that produces a commitment collision, or a good witness $\bm{x'}$ solution of the Rank-SD instance, by making an average number of calls to $\pt$ bounded from above: $$\dfrac{4}{\Tilde{\epsilon}-\varepsilon}\left(1+\dfrac{8\cdot (N-\ell)}{\Tilde{\epsilon}-\epsilon}\right)$$
\end{theorem}

\begin{proof}
For this proof, we refer to \cite{FR22}, which proves this theorem for any MPC protocol fitting this model. Our threshold protocol is an exact application of this model. Hence, the proof of the above theorem is the same as the proof in appendix D of \cite{FR22}.
\end{proof}

\begin{theorem}
    If the algorithm PRG is pseudo-random generator and the commitment $\mathsf{Com}$ is hiding, then there exists an honest-Verifier Zero Knowledge algorithm.
\end{theorem}

The proof is similar to what is done in the case of the additive sharing, and holds by the $t$-privacy of the MPC protocol, as well as the hiding property of the commitments. Moreover, a proof in the general case is provided in [FR22, Appendix C]. Since we are in their model of MPCitH, the proof applies here as well.

\subsection{Security proof for signature schemes} \label{sec_sgn}

As in the previous subsection, the proofs are largely inspired by the work of \cite{AGHHJY22}, which is largely inspired by \cite{FJR22}. The two proofs are very similar, and proceed in two main steps: we begin by explaining how efficiently simulate a signature oracle to a CMA adversary using public key and the honest verifier zero-knowledge simulator. We then show how such a RUF-CMA adversary can be used to extract a solution to the Rank-SD instance.

\subsubsection{Additive-RYDE}

We begin by computing the probability of an attack on the signature based in the hypercube MPC.

\begin{theorem}\label{secu_sig_hypercube}
    Let the signature be $(t,\epsilon_{PRG})$-secure and with adversary of advantage at most $\epsilon_{RSD}$ against the underlying Rank-SD problem. Let $\hash_0$, $\hash_1$, $\hash_2$, $\hash_3$, and $\hash_4$ be random oracles with output length $2\lambda$ bits. If an adversary makes $q_i$ queries to $\hash_i$ and $q_S$ queries to the signature oracle, the probability of such an adversary producing an existential forgery under chosen message attack (EUF-CMA) is bounded from above by:
    $$\prb[\mathsf{Forge}]\:\leq\: \frac{3 \cdot (q+\tau\cdot N \cdot q_S)^2}{2\cdot 2^\lambda}+\frac{q_S \cdot (q_S + 5q) }{2^\lambda}]+\epsilon_{PRG}+\prb[X+Y=\tau]+\epsilon_{RSD}$$
    where $q = \operatorname{max}\{q_0,q_1,q_2,q_3,q_4\}$, $\tau$ is the number of rounds of the signature, $p=\frac{1}{q^{m\eta}}+\left(1-\frac{1}{q^{m\eta}}\right)\frac{1}{q^{m\eta}}$ is the false positive rate of the underlying MPC protocol, $X = \operatorname{max}_{i\in [0,q_2]}\{X_i\}$ with $X_i \sim \mathcal{B}(\tau,p)$, and $Y = \operatorname{max}_{i\in [0,q_4]}\{Y_i\}$ with $Y_i \sim \mathcal{B}(\tau-X,\frac{1}{N})$.
\end{theorem}

\begin{proof}
In this proof, we adopt a game hopping strategy in order to find the upper bound. We note $\prb_i[\mathsf{Forge}]$ the probability of forgery when considering game $i$.

The first game is the access to the standard signing oracle by the adversary $\mathcal{A}$. Then, we game hop in order to eliminate the cases where collisions happen, and, through some other games, we manage to find an upper bound. The aim of the proof is to find this bound on $\prb_1[\mathsf{Forge}]$.\\
\begin{itemize} 
\item\textbf{Game 1:} \\
    This is the interaction between $\mathcal{A}$ and the real signature scheme. \\
    $\mathsf{KeyGen}$ generates $(\bm{H},\bm{y})$ and $\bm{x}$, and $\mathcal{A}$ receives $(\bm{H},\bm{y})$. $\mathcal{A}$ can make queries to each $\hash_i$ independently, and can make signing queries. At the end of the attack, $\mathcal{A}$ outputs a message/signature pair, $(m,\sigma)$. The event $\mathsf{Forge}$ happens when the message output by $\mathcal{A}$ was not previously used in a query to the signing oracle.\\

\item\textbf{Game 2:} \\
    In this game, we add a condition to the success of the attacker. The condition we add is that if there is a collision between outputs of $\hash_0$, or $\hash_1$, or $\hash_3$, then, the forgery isn't valid. \\
    The first step is to look at the number of times every $\hash_i$ is called when calling the signing oracle. For $\hash_0$, we make $\tau \cdot N$ queries. The signing oracle contains also $\tau$ calls to $\hash_1$, one to $\hash_2$, $\tau\cdot D$ to $\hash_3$, and finally, a single one to $\hash_4$. \\
    The number of queries to $\hash_0$ or $\hash_1$ or $\hash_3$ is then bounded from above by $q+\tau \cdot N \cdot q_S$.\\
    We can then have the following result (it comes simply from the probability to have at least one collision with $q+\tau\cdot N\cdot q_S$ values):  
    $$\abs{\prb_1[\mathsf{Forge}]-\prb_2[\mathsf{Forge}]} \le \frac{3 \cdot (q+\tau\cdot N \cdot q_S)^2}{2\cdot 2^\lambda}$$

\item\textbf{Game 3:} \\
    The attacker now fails if the inputs to any of the $\hash_i$ has already appeared in a previous query. If that happens, this means that at least the salt used was the same (we emphasize on \textit{at least}). We have one $\salt$ sampled every time a query is made to the signing oracle, and it can collide each time with: a previous $\salt$, or any of the queries to the $\hash_i$. This means, we can bound this with: $$\abs{\prb_2[\mathsf{Forge}]-\prb_3[\mathsf{Forge}]} \le \frac{q_S \cdot (q_S + q_0 + q_1 + q_2 +q_3 +q_4)}{2^{\lambda}} \le \frac{q_S \cdot (q_S + 5q)}{2^\lambda}$$

\item\textbf{Game 4:} \\
    To answer the signing queries, we now use the \textbf{HVZK} simulator built in the previous proof, in order to generate the views of the open parties. By security of the PRG, the difference with the previous game is:     $$\abs{\prb_7[\mathsf{Forge}]-\prb_6[\mathsf{Forge}]}\le\epsilon_{PRG}$$

\item\textbf{Game 5:} \\
    Finally, we say that an execution $e^*$ of a query $ h_2 = \hash_4(m,\salt,h_1,{H_1^{(e)} \dots H_D^{(e)}}_{e\in \oneto{\tau}})$ defines a good witness $\bm{x}$ if:\begin{itemize}
        \item Each of the $H_k^{(e)}$ are the output of a query to $\hash_3$
        \item $h_1$ is the output of a query to $\hash_2$, i.e.: $$h_1 = \hash_2(\salt,m,h_0^{(1)},\dots , h_0^{(\tau)})$$
        \item Each $h_0^{(e)}$ is the output of a query to $\hash_1$, i.e.: $$h_0^{(e)} = \hash_1(\cmt_1^{(e)} ,\dots, \cmt_{N}^{(e)})$$
        \item Each $\cmt_i^{(e)}$ is the output of a query to $\hash_0$, i.e.: $$\cmt_i^{(e)} = \hash_0(\salt,e,i,\state_i^{(e)})$$
        \item The vector $\bm{x} \in \Fqk$ defined by states $\{\state_i\}_{i \in \oneto{N}}$ is a correct witness, i.e.: $\bm{H}\bm{x} = \bm{y}$ such that $\operatorname{W}_R(\bm{E}) \le r$.
    In the case where an execution like this happens, we are able to retrieve the correct witness from the states $\{\state_i\}_{i \in \oneto{N}}$, and as a consequence, we are able to solve the Rank-SD instance. This means that $\prb_8[\mathsf{Solve}] \le \epsilon_{RSD}$.
    \end{itemize}

Finally, we only need to look at the upper bound of $\abs{\prb_8[\mathsf{Forge} \cap \mathsf{\overline{Solve}}]}$. This probability is upper bounded by the value $$\prb[X+Y=\tau]$$ where $X = \operatorname{max}_{i\in [0,q_2]}\{X_i\}$ with $X_i \sim \mathcal{B}(\tau,p)$, $Y = \operatorname{max}_{i\in [0,q_4]}\{Y_i\}$ with $Y_i \sim \mathcal{B}(\tau-X,\frac{1}{N})$. We explain this bound below:

$\mathsf{Solve}$ does not happen here, meaning that, to have a forgery after a query to $\hash_4$, $\mathcal{A}$ has no choice but to cheat either on the first round or on the second one. \\
    
    \textit{Cheating at the first round.} For any query $Q_2$ to $\hash_2$, we call the output of this query $h_1$. For any query $Q_2$, if a false positive appears in a round $e$ with this value of $h_1$, then we add this round $e$ to the set we call $G_2(Q_2,h_1)$. This means that $\prb(e \in G_2(Q_2,h_1)\:|\:\mathsf{\overline{Solve}}) \le p$. Since the response $h_1$ is uniformly sampled, each round $e$ has the same probability to be in the set $G_2(Q_2,h_1)$. This means that $\#G_2(Q_2,h_1)$ follows the binomial distribution $X_{Q_2} = \mathcal{B}(\tau,p)$. We can then define $(Q_{2\text{best}},h_{1\text{best}})$ such that $\#G_2(Q_2,h_1)$ is maximized, i.e, $$\#G_2(Q_{2\text{best}},h_{1\text{best}}) \sim X = {\operatorname{max}\{X_{Q_2}\}}_{(Q_2 \in \mathcal{Q}_2)}$$

    \textit{Cheating at the second round.} Now, we need to look at the cheating in the second round, i.e, the queries to $\hash_4$. We note this query $Q_4$, with the output of this query $h_2$.
    For the signature to be accepted, we know that, if in a round, the prover sends a wrong value of $h_1$, then he needs to cheat on exactly one party (it is already established that is isn't possible to cheat on less, or on more, than one party). He only needs to cheat when the value of $h_1^{(e)}$ is wrong, i.e, he needs to cheat for every round $e \notin G_2(Q_{2\text{best}},h_{1\text{best}})$. Since every time he cheats, the probability to be detected is $\frac{1}{N}$, it is easy to see the probability that the verification outputs ACCEPT is upper bounded by $\Big( \frac{1}{N}\Big)^{\tau-\#G_2(Q_{2\text{best}},h_{1\text{best}})} $ \\
    The probability that the prover is accepted on one of the $q_4$ queries is then upper bounded by \begin{displaymath}
         1-\Bigg(1- \Big( \frac{1}{N}\Big)^{\tau-\tau_1} \Bigg)
    \end{displaymath}
    where $\tau_1 =\#G_2(Q_{2\text{best}},h_{1\text{best}})$.
    By summing over all values of $\tau_1$ possible, we have then the upper bound: $$\prb_8(\mathsf{Forge} \cap \overline{\mathsf{Solve}}) \le \prb(X+Y = \tau)$$ where $X$ is as before, and $Y = {\operatorname{max}\{Y_{Q_2}\}}_{(Q_2 \in \mathcal{Q}_2)}$ where the $Y_{Q_2}$ are distributed following $\mathcal{B}(\tau-X,\frac{1}{N})$. 
    \end{itemize}
    All that is left to do is then to compute the sum of all the upper bounds we retrieved: this gives us the wanted result.
\end{proof}

\subsubsection{Threshold-RYDE}

We compute now the probability to forge the signature based on the MPC threshold.

\begin{theorem}\label{secu_sig_th}
    Let the signature be $(t,\epsilon_{PRG})$-secure and with adversary of advantage at most $\epsilon_{RSD}$ against the underlying Rank-SD problem. Let $\hash_0$, $\hash_1$, $\hash_2$ and $\hash_M$ be random oracles with output length $2\lambda$ bits. The probability of such an adversary producing an existential forgery under chosen message attack (EUF-CMA) is bounded from above by:$$\prb[\mathsf{Forge}]\:\leq\: \frac{\cdot (q+\tau\cdot (2N-1) \cdot q_S)^2}{2^\lambda}+\frac{q_S \cdot (q_S + 3q)}{2^\lambda}+q_S\cdot\tau\cdot\epsilon_{PRG}+\prb[X+Y=\tau] +\epsilon_{RSD}$$
    where $\tau$ is the number of rounds of the signature, $p=\frac{1}{q^{m\eta}}+\left(1-\frac{1}{q^{m\eta}}\right)\frac{1}{q^{m\eta}}$ the false positive rate of the MPC protocol, $X = \operatorname{max}_{i\in [0,q_1]}\{X_i\}$ with $X_i \sim \mathcal{B}(\tau,\binom{N}{\ell+1}\cdot p)$, and $Y = \operatorname{max}_{i\in [0,q_2]}\{Y_i\}$ with $Y_i \sim \mathcal{B}\left(\tau-X,\frac{1}{\binom{N}{\ell}}\right)$.
\end{theorem}

\begin{proof}
The proof is based on the same operation as the previous one, adopting a game hopping strategy. The first game is the access to the standard signing oracle by the adversary $\mathcal{A}$. The aim of the proof is to find this bound on $\prb_1[\mathsf{Forge}]$.\\

\begin{itemize}

\item\textbf{Game 1:} \\
    This is the interaction between $\mathcal{A}$ and the real signature scheme. \\
    $\mathsf{KeyGen}$ generates $(\bm{H},\bm{y})$ and $\bm{x}$, and $\mathcal{A}$ receives $(\bm{H},\bm{y})$. $\mathcal{A}$ can make queries to each $\hash_i$ independently, and can make signing queries. At the end of the attack, $\mathcal{A}$ outputs a message/signature pair, $(m,\sigma)$. The event $\mathsf{Forge}$ happens when the message output by $\mathcal{A}$ was not previously used in a query to the signing oracle.\\

\item\textbf{Game 2:} \\
    We add a condition to the success of the attacker now. If there is a collision in the outputs of $\hash_0$ or on $\hash_M$, then the forgery isn't valid. Here, $\hash_0$ is called $q_0$ times by $\mathcal{A}$, $\hash_M$ $q_M$ times. When $\mathcal{A}$ calls the signing oracle, there are in total: $\tau \cdot N$ calls to $\hash_0$, $\tau \cdot (2N-1)$ calls to $\hash_M$, and one to $\hash_1$ and $\hash_2$. In this game, only $\hash_0$ and $\hash_M$ are of interest. We can then give an upper bound to the queries made by $\mathcal{A}$ to the hash functions, which is then: $q + \tau \cdot (2N-1)\cdot q_S$ where $q_S$ is the number of queries to the signing oracle. \\
    When making this many queries, we can now bound from above the probability of having a collision, with  $$\abs{\prb_1[\mathsf{Forge}]-\prb_2[\mathsf{Forge}]} \le \frac{(q + \tau \cdot (2N-1)\cdot q_S)^2}{2^\lambda}$$

\item\textbf{Game 3:} \\
    The attacker now fails if the inputs to any of the hash functions has already appeared in a previous query. If that happens, this means that at least the $\salt$ used was the same. Since we don't use $\salt$ in the Merkle Tree, this only concerns $\hash_0,\hash_1,\hash_2$. We sample $\salt$ $q_S$ time (once by signing oracle query), and $3q$ times as well (each time we call $\hash_0,\hash_1$ or $\hash_2$). If there is an input which already appears for $\hash_M$, this must be because a collision has been found either on $\hash_M$ or on $\hash_0$. However, we already excluded this in \textbf{Game 2}. This means we can give the following bound: 
    $$\abs{\prb_2[\mathsf{Forge}]-\prb_3[\mathsf{Forge}]} \le \frac{q_S \cdot (q_S + q_0 + q_1 + q_2)}{2^{\lambda}} \le \frac{q_S \cdot (q_S + 3q)}{2^\lambda}$$

\item\textbf{Game 4:} \\
    To answer the signing queries, we now use the \textbf{HVZK} simulator built in the previous proof, in order to generate the views of the open parties. By security of the PRG, the difference with the previous game is:     $$\abs{\prb_7[\mathsf{Forge}]-\prb_6[\mathsf{Forge}]}\le\epsilon_{PRG}$$

\item\textbf{Game 5:} \\
    Finally, we say that an execution $e^*$ of a query $ h_2 = \hash_2(m,\mathsf{pk},\salt,h_1, (\share{\boldsymbol{\alpha}^{(e)}}_i, \share{v^{(e)}}_i)_{i \in S, e \in \oneto{\tau}})$ defines a good witness $\bm{x}$ if: \begin{itemize}
        \item $h_1$ is the output of a query to $H_1$, i.e, $$h_1 = \hash_2(\salt,m,h_0^{(1)},\dots , h_0^{(\tau)})$$
        \item Each $h_0^{(e)}$ is the output of a query to the MerkleTree oracle, i.e, $$h_0^{(e)} = \mathsf{Merkle}(\cmt_1^{(e)} ,\dots, \cmt_{N}^{(e)})$$
        \item Each $\cmt_i^{(e)}$ is the output of a query to $\hash_0$, i.e, $$\cmt_i^{(e)} = \hash_0(\salt,e,i,\state_i^{(e)})$$
        \item The vector $\bm{x} \in \Fqk$ defined by states $\{\state_i\}_{i \in S}$ is a correct witness, i.e.: $\bm{H}\bm{x} = \bm{y}$ such that $\operatorname{W}_R(\bm{E}) \le r$.
    \end{itemize}
    In the case where an execution like this happens, we are able to retrieve the correct witness from the states $\{\state_i\}_{i \in \oneto{N}}$, and as a consequence, we are able to solve the Rank-SD instance. This means that $\prb_8[\mathsf{Solve}] \le \epsilon_{RSD}$.

    Finally, we only need to look at the upper bound of $\abs{\prb_8[\mathsf{Forge} \cap \mathsf{\overline{Solve}}]}$.
    This probability is upper bounded by the value $$\prb[X+Y=\tau]$$ with $X = \operatorname{max}_{i\in [0,q_2]}\{X_i\}$ with $X_i \sim \mathcal{B}(\tau,\binom{N}{\ell+1}\cdot p)$, $Y = \operatorname{max}_{i\in [0,q_4]}\{Y_i\}$ with $Y_i \sim \mathcal{B}(\tau-X,\frac{1}{\binom{N}{\ell}})$ and where $p=\frac{1}{q^{m\eta}}+\left(1-\frac{1}{q^{m\eta}}\right)\frac{1}{q^{m\eta}}$. \\
    This result comes directly from \cite[Lemma 6 and Theorem 4, Appendix F]{FR22}.

\end{itemize}

All that is left to do is then to compute the sum of all the upper bounds we retrieved: this gives us the wanted result.

\end{proof}

\section{Known Attacks}

\subsection{Attacks against Fiat-Shamir signatures} \label{atk_fs}

There are several attacks against signatures from zero-knowledge proofs obtained thanks to the Fiat-Shamir heuristic. \cite{AAB} propose an attack more efficient than the brute-force one for protocols with more than one challenge, i.e. for protocols of a minimum of 5 rounds.

Kales and Zaverucha proposed in \cite{KZ20} a forgery achieved by guessing separately the two challenges of the protocol. It results an additive cost rather than the expected multiplicative cost. The cost associated with forging a transcript that passes the first 5 rounds of the Proof of Knowledge (Figure\ref{pzk_hypercube} or Figure\ref{pzk_thr}) relies on achieving an optimal trade-off between the work needed for passing the first step and the work needed for passing the second step. To achieve the attack, one can find an optimal number of repetitions with the formula: $$\tau_1=\arg\min_{0\leq\tau\leq r}\left\{\frac{1}{P_1}+P_2^{r-\tau_1}\right\}$$ where $P_1$ and $P_2$ are the probabilities to pass respectively the first and the second challenge.

In the case of the additive case, one obtains:
$$\text{cost}_{\text{forge}}=\min_{0\leq\tau'\leq\tau}\left\lbrace\dfrac{1}{\sum_{i=\tau'}^\tau \binom{\tau}{i}p^i(1-p)^{\tau-i}}+(N)^{\tau-\tau'}\right\rbrace$$
with $p = \frac{2}{q^{m\eta}}-\frac{1}{q^{2m\eta}}$. 

In the case of low-threshold case, one obtains:
$$\text{cost}_{\text{forge}}=\min_{0\leq\tau'\leq\tau}\left\lbrace\dfrac{1}{\sum_{i=\tau'}^\tau \binom{\tau}{i}p'^i(1-p')^{\tau-i}}+\binom{N}{\ell}^{\tau-\tau'}\right\rbrace \label{kz_eq_th_1}$$
with $p' = \Big(\frac{2}{q^{m\eta}}-\frac{1}{q^{2m\eta}}\Big) \binom{N}{\ell+1}$.

\subsection{Attacks against Rank-SD problem}\label{atk_rsd}

We give the complexity of the main attacks against the Rank-Syndrome-Decoding problem. The interested reader can found the detailed attacks and proofs in \cite{specs}.

The enumeration of basis is a combinatorial attack which consists in trying all the different possible supports fof the error. The complexity of this attack is upper bounded by: $$O((nr+m)^3q^{(m-r)(r-1)})$$.

An other combinatorial attack is the error support attack, which is the adaptation of the information set decoding attack used in Hamming metric. It consists in guessing a set of $n-k$ coordinates which contains the support of the vector $\bm{x}$ to obtain a system of $n-k$ equations with $n-k$ variables, for which there exists often a solution. This attack can be achieved with an average complexity: $$O((n-k)^3m^3q^{(r-1)\lfloor\frac{(k+1)m}{n}\rfloor})$$

There exists also algebraic attacks to solve the system of equations $\bm{H}\bm{x}=\bm{y}$ using computer algebra techniques like Gröbner basis. Today, the best algebraic modelings for solving the Rank-SD problem are the MaxMinors modeling~\cite{BBBGNRT20,BBCGPSTV20} and the Support Minors modeling~\cite{BBCGPSTV20,BBBGT22}. The final cost in $\Fq$ operations is given by:
  \begin{align*}
    \mathcal{O}\left( m^2N {M }^{\omega-1}\right),
  \end{align*} where \begin{align*}
  N &=  \sum_{i = 1}^{k} \binom{n-i}{ r}\binom{k+b-1-i}{b-1} 
      - \binom{n-k-1}{ r}\binom{k+b-1}{b} \\
  &~ -  (m-1)\sum_{i=1}^{b} (-1)^{i+1} 
    \binom{k+b-i-1}{b-i}\binom{n-k-1}{r+i}.\\
  M &= \binom{k+b-1}{b}\left(\binom{n}{r} - m\binom{n-k-1}{r} \right).
\end{align*} and $\omega$ is the linear algebra constant.

\newpage

\bibliographystyle{alpha}
\bibliography{biblio}

\nocite{isit}

\end{document}

%% file: fig-mpc_rsd_opt.tex
\begin{figure}[h!]
    \pcb[codesize=\scriptsize, minlineheight=0.75\baselineskip, mode=text, width=0.98\textwidth] { 
  \textbf{Public values.} instances of a Rank-SD problem : $\bm{H} = (I ~ \bm{H}') \in \Fqm^{\nmktn}$ and $\bm{y}\in \Fqm^{\nmk}$\\
  \textbf{Inputs.} Each party takes a share of $\share{\bm{x}_B}$ where $\bm{x}\in\Fqmn$. Let $U$ a $\Fq$-linear subspace of $\Fqm$ of dimension $r$ which contains $x_1,\dots,x_n$ and has for annihilator polynomial $L=\prod_{u\in U}(X-u)$. Each party takes a share of the following sharing as inputs : $\share{L}=\sum_{i=1}^{r-1}\share{\beta_i}(X^{q^i}-X)+(X^{q^r}-X)$, $\share{\bm{a}}$ where $\bm{a}$ is uniformly sampled from $\Fqme^{r-1}$ and $\share{c}$ such that $c=-\langle\bm{\bm{\beta}},\bm{a}\rangle$.\\[\baselineskip]
  \textbf{MPC Protocol :}\\[\baselineskip]
  1. The parties sample uniformly random : $(\gamma_1,...,\gamma_n,\varepsilon)\sampler\Fqme^{n+1}$.\\
  2. The parties locally compute $\share{\bm{x}_A}=\bm{y}-\bm{H}'\share{\bm{x}_B}$ and $\share{\bm{x}} = (\share{\bm{x}_A} \,||\, \share{\bm{x}_B})$.\\
  3. The parties locally compute $\share{z} = -\sum_{j=1}^n\gamma_j\left(\share{x_j}^{q^r}-\share{x_j}\right)$.\\
  4. The parties locally compute for all $k$ from 1 to $r-1$ : $\share{\omega_k} = \sum_{j=1}^n\gamma_j\left(\share{x_j}^{q^k}-\share{x_j}\right)$.\\
  5. The parties locally compute $\share{\bm{\alpha}}=\varepsilon\cdot \share{\bm{\omega}} + \share{\bm{a}}$.\\
  6. The parties open $\bm{\alpha}$.\\
  7. The parties locally compute $\share{v} = \varepsilon\cdot \share{z} - \langle\bm{\alpha},\share{\bm{\beta}}\rangle-\share{c}$.\\
  8. The parties open $v$.\\
  9. The parties output ACCEPT if $v=0$, and reject otherwise. 
  }
\vspace{-\baselineskip}
\captionof{figure}{\footnotesize{Protocol $\Pi^\eta$ : checking that an input is solution of an instance of Rank-SD problem}}
\label{mpc_opti}
\end{figure}

%% file: fig-pzk_rsd_hyp.tex
\begin{figure}[h!]
    \pcb[codesize=\scriptsize, minlineheight=0.75\baselineskip, mode=text, width=0.98\textwidth] { 
  \textbf{Public data}: An instance of a Rank-SD problem: $\bm{H} = (I ~ \bm{H}') \in \Fqm^{\nmktn}$ and $\bm{y}\in \Fqm^{\nmk}$\\
  The prover wants to convince the verifier that he knows the solution $\bm{x}\in\Fqmn$ of the instance, i.e. such that $\bm{y} = \bm{H}\bm{x}$. \\[\baselineskip]
  \textbf{Step 1: Commitment} \\
  1. The prover computes $\bm{\beta} = (\bm{\beta}_k)_{k \in [1, r - 1]} \in \Fqm^r$ the coefficients of the annihilator q-polynomial $L(X)$ associated to $\bm{x}$ such that $L(X)=\prod_{u\in U}(X-u) = (X^{q^r} - X) + \sum_{k = 1}^{r - 1} \bm{\beta}_k (X^{q^k} - X)$ and $\forall j \in \oneto{n}, L(\bm{x}_j) = 0$ \\
  2. The prover samples a seed $\seed\sampler\{0, 1\}^\lambda$ and a random salt $\salt\sampler\{0, 1\}^\lambda$.\\
  3. The prover expand $\seed$ recursively using a GGM tree with $\salt$ to obtain $N$ leaves and seeds $(\seed_{i'},\rho_{i'})$.\\
  4. For each $i\in\oneto{N-1}$: \\
  \pcind \pcind - The prover samples $(\share{\bm{x}_B}_i,\share{\bm{\beta}}_i,\share{\bm{a}}_i\share{c}_i)\samples{\seed_i}PRG$ where PRG is a pseudo-random generator \\
  \pcind \pcind - $\state_i=\seed_i$\\
  5. For the share $N$: \\
  \pcind \pcind - The prover samples $\share{\bm{a}}_{N}\samples{\seed_{N}} PRG$\\
  \pcind \pcind - The prover computes $\share{\bm{x}_B}_{N}=\bm{x}_B-\sum_{i=1}^{N-1}\share{\bm{x}_B}_i$, $\share{\bm{\beta}}_{N}=\bm{\beta}-\sum_{i=1}^{N-1}\share{\bm{\beta}}_i$ and $\share{c}_{N}=-\dotp{\bm{a}, \bm{\beta}}-\sum_{i=1}^{N-1}\share{c}_i$\\
  \pcind \pcind - $\state_{N}=(\seed_{N},\share{\bm{x}_B}_{N},\share{\bm{\beta}}_{N},\share{c}_{N})$\\
  6. The prover computes the commitments: $\cmt_i=\commit{\state_i,\rho_i}$, where $\rho_i$ is a random tape sampled in $\{0, 1\}^\lambda$ from $\seed_i$ for each $i\in\oneto{N}$.\\
  7. The prover commits all the shares: $h_0 = \hash(\cmt_1, \cdots, \cmt_{N})$\\
  8. The prover computes the main parties by summing all the leaves associated. Indexing each party by its coordinates on the hypercube: $i=(i_1,...,i_D)$ where $i_k\in\oneto{2}$. For all main party index $p=(d,k)\in\lbrace (1,1),...(D,2)\rbrace$: $\share{\bm{x}_B}_{(d,k)}=\sum_{i: i_{d}=k}\share{\bm{x}_B}_i$, $\share{\bm{\beta}}_{(d,k)}=\sum_{i: i_{d}=k}\share{\bm{\beta}}_i$, $\share{\bm{a}}_{(d,k)}=\sum_{i: i_{d}=k}\share{\bm{a}}_i$ and $\share{c}_{(d,k)}=\sum_{i: i_{d}=k}\share{c}_i$.\\
  [\baselineskip]
  \textbf{Step 2: First Challenge} \\
  9. The verifier randomly samples $\big( (\gamma_j)_{j \in \oneto{n}}, \epsilon \big) \in \Fqme^n \times \Fqme$ and sends it to the prover.\\[\baselineskip]
  \textbf{Step 3: First Response} \\
  10. For each dimension $d\in\oneto{D}$, the prover executes the algorithm in Fig. \ref{exec_pi} on the main parties of the current dimension: $(d,1),(d,2)$. He computes $(\share{\bm{\alpha}},\share{v},H_d)\longleftarrow \text{Algorithm\ref{exec_pi}}(\big( P_d,(\gamma_j)_{j \in \oneto{n}}, \epsilon \big)$\\
  11. The prover commits the executions: $h_1=\mathsf{H}(H_1,...,H_D)$\\[\baselineskip]
  \textbf{Step 4: Second Challenge}\\
  12. The verifier randomly samples $i^*\in\oneto{N}$ and sends it to the prover.\\[\baselineskip]
  \textbf{Step 4: Second Response and verification}\\
  13. The prover sends to the verifier $\cmt_{i^*}$ and $\share{\bm{\alpha}}_{i^*}$; and the sibling path of the share $i^*$ to retrieve $(\state_j,\rho_j)_{j\neq i^*}$.\\
  14. The verifier can deduce all the leaves (with the exception of index $i^*$), and recover $h_0$ using the sibling path and $\cmt_{i^*}$.\\
  15. For all dimension $d\in\oneto{D}$, the verifier runs the algorithm in Fig. \ref{check_pi} to get $\share{\bm{\alpha}}$, $\share{v}$ and $H_d$. He checks:\\
  \pcind \pcind - $v=0$ for all the executions.\\
  \pcind \pcind - $\bm{\alpha}$ is the same for all $D$ main executions.\\
  \pcind \pcind -  $h_1=\mathsf{H}(H_1,...,H_D)$
  }
\vspace{-\baselineskip}
\captionof{figure}{\footnotesize{Protocol optimized with the hypercube technique}}
\label{pzk_hypercube}
\end{figure}

%% file: fig-exec_pi.tex
\begin{figure}[h!]
    \centering
    \pcb[codesize=\scriptsize, minlineheight=0.75\baselineskip, mode=text, width=0.98\textwidth] { 
  \textbf{Inputs :} A set of shares for the dimension $k$: $(\share{\bm{x}_B},\share{\bm{\beta}},\share{\bm{a}},\share{c})$ and a protocol challenge $\big( (\gamma_j)_{j \in \oneto{n}}, \epsilon \big)$\\
   \textbf{Outputs :} A set of shares $\share{v}$ and a commitment $H$ of the execution\\[\baselineskip]
  For each party $i\in\oneto{N}$ : \\
 \pcind \pcind - Compute $\share{\bm{x}_A}_i = \bm{y} - \bm{H}' \share{\bm{x}_B}_i$ and $\share{\bm{x}}_i = (\share{\bm{x}_A}_i \,||\, \share{\bm{x}_B}_i)$ \\
  \pcind \pcind - $\share{z}_i = -\sum_{j=1}^n\gamma_j\left(\share{x_j}_i^{q^r}-\share{x_j}_i\right)$ \\
  \pcind \pcind -  Compute $\share{\omega_k} = \sum_{j=1}^n\gamma_j\left(\share{x_j}^{q^k}-\share{x_j}\right)$ for each $k \in [1, r - 1]$ \\
 \pcind \pcind - Compute $\share{\boldsymbol{\alpha}}_i = \epsilon \cdot \share{\bm{w}}_i + \share{\bm{a}}_i$ and reveals $\boldsymbol{\alpha}$\\
 \pcind \pcind -  Compute $\share{v}_i = \epsilon \cdot \share{z}_i - \dotp{\boldsymbol{\alpha}, \share{\bm{\beta}}_i} - \share{c}_i$ \\
 The parties compute together : $H=\hash \left(\share{\boldsymbol{\alpha}},\share{v})\right)$
  }
\vspace{-\baselineskip}
\captionof{figure}{\footnotesize{Execution of the MPC protocol $\Pi^\eta$ on a set of main shares}}
\label{exec_pi}
\end{figure}

%% file: fig-check_pi.tex
\begin{figure}[h!]
\pcb[codesize=\scriptsize, minlineheight=0.75\baselineskip, mode=text, width=0.98\textwidth] { 
  \textbf{Inputs:} A leaf $i^*$ that one does not reveal, the main party shares $\share{\bm{\alpha}}$ and $\share{v}$ on which the hidden leaf $i^*$ depends on, all the other main parties shares $(\share{\bm{x}_B},\share{\bm{\beta}},\share{\bm{a}},\share{c})$. \\
   \textbf{Outputs:}  $\share{\bm{\alpha}},\share{v}$ and a commitment $H$ of the execution\\[\baselineskip]
  \textbf{For} all the main parties:\\
  \pcind \pcind \textbf{If} the party $p=(d,k)$ contains the leaf $i^*$, i.e. $i^*_d=k$:\\
  \pcind \pcind \pcind \pcind Set $\share{v}_{i^*}$ such that $v=0$\\
  \pcind \pcind \textbf{Else}:\\
  \pcind \pcind \pcind \pcind Do the same computations as in $\Pi^\eta$ to obtain the correct shares $\share{\bm{\alpha}}_p$ and $\share{v}_p$\\
  Compute $H=\hash \left(\share{\boldsymbol{\alpha}},\share{v}\right)$
  }
\vspace{-\baselineskip}
\captionof{figure}{\footnotesize{Check the executions of the MPC protocol $\Pi^\eta$ on a the main parties}}
\label{check_pi}
\end{figure}

%% file: fig-sign_hyp.tex
\begin{figure}[h!]
    \pcb[codesize=\scriptsize, minlineheight=0.75\baselineskip, mode=text, width=0.98\textwidth] { 
  \textbf{Inputs} \\
  \pcind - Secret key $\sk = \bm{x}_B$ with $\bm{x} = (\bm{x}_A \,||\, \bm{x}_B) \in \Fqmn$ such that $\rw{\bm{x}} = r$ \\
  \pcind - Public key $\pk = (\bm{H}, \bm{y})$ with $\bm{H} = (\bm{I} ~ \bm{H}') \in \Fqm^{\nmktn}$ and $\bm{y} \in \Fqm^{\nmk}$ such that $\bm{y} = \bm{H}\bm{x}$ \\
  \pcind - Message $m \in \{0, 1\}^*$ \\[\baselineskip]
  \textbf{Step 1: Commitment} \\
  1. Sample a random salt value $\salt \sampler \{0, 1\}^{2\lambda}$ \\
  2. Compute $\bm{\beta} = (\bm{\beta}_k)_{k \in [1, r - 1]} \in \Fqm^r$ the coefficients of the annihilator q-polynomial $L(X)$ associated to $\bm{x}$ such that $L(X) = \prod_{u\in U}(X-u) = (X^{q^r} - X) + \sum_{k = 1}^{r - 1} \bm{\beta}_k (X^{q^k} - X)$ and $\forall j \in \oneto{n}, L(\bm{x}_j) = 0$ \\
  3. For each iteration $e \in \oneto{\tau}$: \\
  \pcind - Sample a root seed $\seed^{(e)}\leftarrow \lbrace 0,1\rbrace^\lambda$.\\
  \pcind - Expand $\seed^{(e)}$ recursively using a GGM tree with $\salt$ to obtain $N$ seeds and randomness $(\seed^{(e)}_{i'},\rho_{i'})$ \\
  \pcind \textbf{For} each $i\in\oneto{N}$ : \\
  \pcind \pcind \pcind \textbf{If} $i'\neq N$ :\\
  \pcind \pcind \pcind \pcind \pcind - Sample $(\share{\bm{x}_B^{(e)}}_i,\share{\bm{\beta}^{(e)}}_i,\share{\bm{a}^{(e)}}_i\share{c^{(e)}}_i)\samples{\seed_i^{(e)}} PRG$ where PRG is a pseudo-random generator \\
  \pcind \pcind \pcind \pcind \pcind - $\state_i^{(e)}=\seed_i^{(e)}$\\
  \pcind \pcind \pcind \textbf{Else} :\\
  \pcind \pcind \pcind \pcind \pcind - Sample $\share{\bm{a^{(e)}}}_{N}\samples{\seed_{N}^{(e)}} PRG$\\
  \pcind \pcind \pcind \pcind \pcind - Compute $\share{\bm{x}_B^{(e)}}_{N}=\bm{x}_B^{(e)}-\sum_{i=1}^{N-1}\share{\bm{x}_B^{(e)}}_i$, $\share{\bm{\beta}^{(e)}}_{N}=\bm{\beta}^{(e)}-\sum_{i=1}^{N-1}\share{\bm{\beta}^{(e)}}_i$ and $\share{c^{(e)}}_{N}=-\dotp{\boldsymbol{\alpha}^{(e)}, \bm{\beta}^{(e)}}-\sum_{i=1}^{N-1}\share{c^{(e)}}_i$\\
  \pcind \pcind \pcind \pcind \pcind - $\state_{N}^{(e)}=(\seed_{N}^{(e)},\share{\bm{x}_B^{(e)}}_{N},\share{\bm{\beta}^{(e)}}_{N},\share{c^{(e)}}_{N})$\\
  \pcind \pcind \pcind - Compute $\cmt_i^{(e)}=\mathsf{H}_0(\salt,e,\state_i^{(e)})$\\
  4. For all main party index $p=(d,k)\in\lbrace (1,1),...(D,2)\rbrace$ : $\share{\bm{x}_B^{(e)}}_{(d,k)}=\sum_{i: i_{d}=k}\share{\bm{x}_B^{(e)}}_i$, $\share{\bm{\beta}^{(e)}}_{(d,k)}=\sum_{i: i_{d}=k}\share{\bm{\beta}^{(e)}}_i$, $\share{\bm{a}^{(e)}}_{(d,k)}=\sum_{i: i_{d}=k}\share{\bm{a}^{(e)}}_i$ and $\share{c^{(e)}}_{(d,k)}=\sum_{i: i_{d}=k}\share{c^{(e)}}_i$.\\
  5. Commit all the shares : $h_0^{(e)}=\mathsf{H}_1(\salt,e,\cmt_1^{(e)},...,\cmt_{N}^{(e)})$\\
  6. Compute $h_1=\mathsf{H}_2\left(\salt,m,h_0^{(1)},...,h_0^{(\tau)}\right)$\\
  [\baselineskip]
  \textbf{Step 2: First Challenge} \\
  7. Extend hash $\big( (\gamma^{(e)}_j)_{j \in \oneto{n}}, \epsilon^{(e)} \big)_{e \in \oneto{\tau}} \leftarrow PRG(h_1)$ where $\big( (\gamma^{(e)}_j)_{j \in \oneto{n}}, \epsilon^{(e)} \big)_{e \in \oneto{\tau}}\in (\Fqme^n \times \Fqme)^{\tau}$ \\
  [\baselineskip]
  \textbf{Step 3: First Response} \\
  8. For each iteration $e \in \oneto{\tau}$ : \\
  \pcind\pcind \textbf{For} each each dimension $k\in\oneto{D}$ :\\
  \pcind \pcind \pcind - Execute the algorithm in Fig. \ref{exec_pi} to obtain $\share{\bm{\alpha}^{(e)}}$, $\share{v^{(e)}}$ and $H_k^{(e)}$\\
  9. Compute $h_2 = \hash_4(m, \pk, \salt, h_1, (H_1^{(e)},...,H_D^{(e)})_{e \in \oneto{\tau}})$\\
  [\baselineskip]
  \textbf{Step 4: Second Challenge} \\
  10. Extend hash $(i^{*(e)})_{e \in \oneto{\tau}}\leftarrow PRG(h_2)$ with $i^{*(e)}\in\oneto{N}$ \\
  [\baselineskip]
  \textbf{Step 5: Second Response} \\
  11.For each iteration $e \in \oneto{\tau}$: \\
  \pcind Compute $\rsp^{(e)}=\left( (\state_j)_{j\neq i^*},\cmt_{i^{*(e)}},\share{\alpha^{(e)}}_{i*^{(e)}}\right)$\\
  12. Output $\sigma = \left(\salt, h_1, h_2, (\rsp^{(e)})_{e \in \oneto{\tau}}\right)$
}
\vspace{-\baselineskip}
\captionof{figure}{\footnotesize{Rank-SD Signature Scheme based on Hypercube MPCitH - Signature Algorithm}}
\label{sign_hyp}
\end{figure}

%% file: fig-verify_hyp.tex
\begin{figure}[h!]
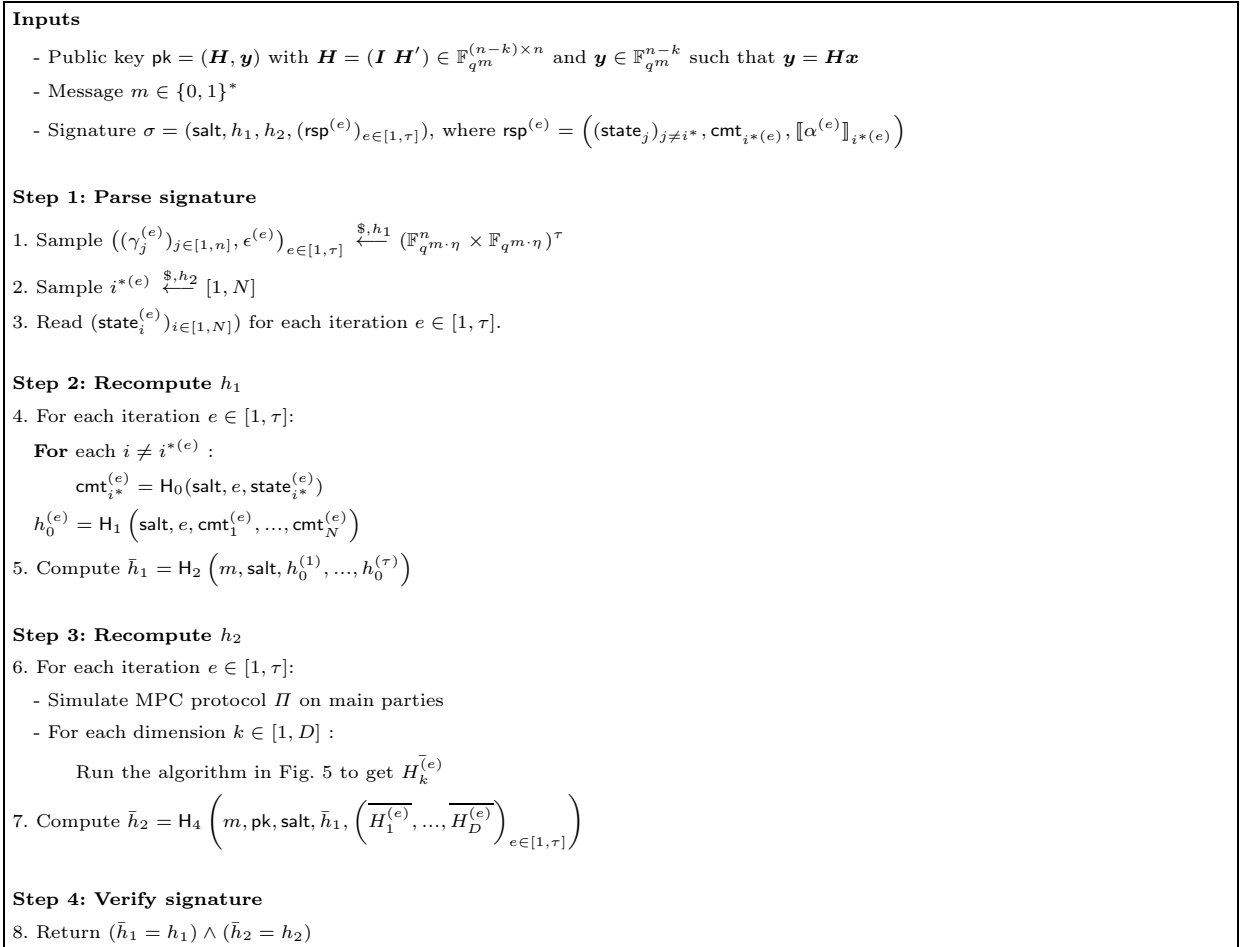

    \pcb[codesize=\scriptsize, minlineheight=0.75\baselineskip, mode=text, width=0.98\textwidth] { 
  \textbf{Inputs} \\
  \pcind - Public key $\pk = (\bm{H}, \bm{y})$ with $\bm{H} = (\bm{I} ~ \bm{H}') \in \Fqm^{\nmktn}$ and $\bm{y} \in \Fqm^{\nmk}$ such that $\bm{y} = \bm{H}\bm{x}$ \\
  \pcind - Message $m \in \{0, 1\}^*$ \\
  \pcind - Signature $\sigma = (\salt, h_1, h_2, (\rsp^{(e)})_{e \in \oneto{\tau}})$, where $\rsp^{(e)}=\left( (\state_j)_{j\neq i^*},\cmt_{i^{*(e)}},\share{\alpha^{(e)}}_{i^{*(e)}}\right)$ \\
  [\baselineskip]
  \textbf{Step 1: Parse signature} \\
  1. Sample $\big( (\gamma^{(e)}_j)_{j \in \oneto{n}}, \epsilon^{(e)} \big)_{e \in \oneto{\tau}} \samples{h_1} (\Fqme^n \times \Fqme)^{\tau}$ \\
  2. Sample $i^{*(e)}\samples{h_2} [1,N]$\\
  3. Read $(\state_i^{(e)})_{i\in\oneto{N}})$ for each iteration $e\in\oneto{\tau}$.\\
  [\baselineskip]
  \textbf{Step 2: Recompute $h_1$}\\
  4. For each iteration $e \in \oneto{\tau}$: \\
  \pcind \textbf{For} each $i\neq i^{*(e)}$ :\\
  \pcind \pcind \pcind $\cmt_{i^*}^{(e)}=\hash_0(\salt,e,\state_{i^*}^{(e)})$\\
  \pcind $h_0^{(e)}=\hash_1\left(\salt,e,\cmt_1^{(e)},...,\cmt_{N}^{(e)}\right)$\\
  5. Compute $\bar{h}_1=\hash_2\left(m,\salt,h_0^{(1)},...,h_0^{(\tau)}\right)$\\
  [\baselineskip]
  \textbf{Step 3: Recompute $h_2$}\\
  6. For each iteration $e \in \oneto{\tau}$: \\
  \pcind - Simulate MPC protocol $\Pi$ on main parties\\
  \pcind - For each dimension $k\in\oneto{D}$ :\\
  \pcind \pcind \pcind Run the algorithm in Fig. \ref{check_pi} to get $\bar{H_k^{(e)}}$\\
  7. Compute $\bar{h}_2 = \hash_4\left(m, \pk, \salt, \bar{h}_1, \left(\overline{H_1^{(e)}},...,\overline{H_D^{(e)}}\right)_{e \in \oneto{\tau}}\right)$  \\[\baselineskip]
  \textbf{Step 4: Verify signature} \\
  8. Return $(\bar{h}_1 = h_1) \wedge (\bar{h}_2 = h_2)$
}
\vspace{-\baselineskip}
\captionof{figure}{\footnotesize{Rank-SD based signature scheme - Verification algorithm}}
\label{ver_hyp}
\end{figure}

%% file: fig-pzk_rsd_thr.tex
\begin{figure}[h!]
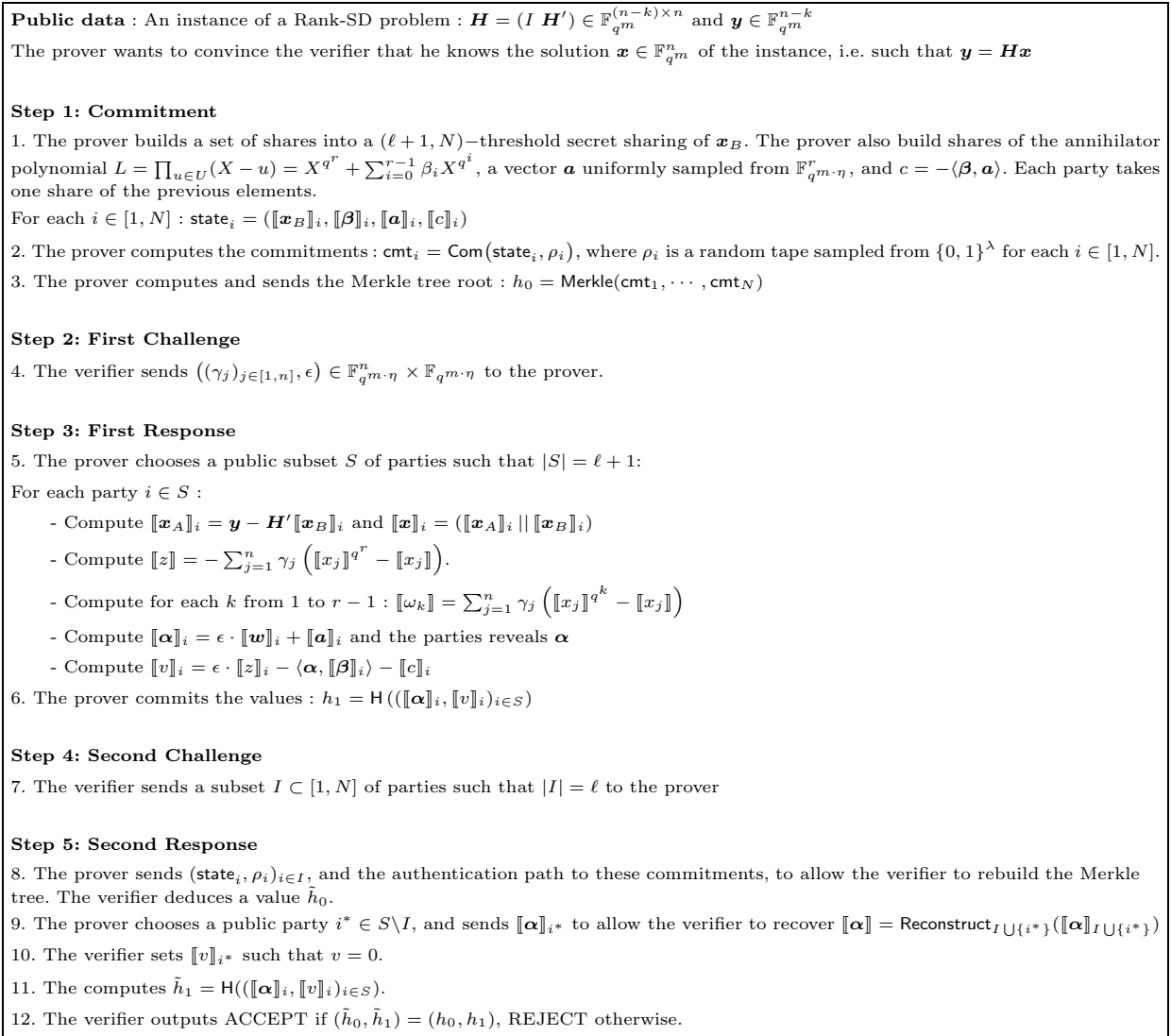

    \pcb[codesize=\scriptsize, minlineheight=0.75\baselineskip, mode=text, width=0.98\textwidth] { 
  \textbf{Public data} : An instance of a Rank-SD problem : $\bm{H} = (I ~ \bm{H}') \in \Fqm^{\nmktn}$ and $\bm{y}\in \Fqm^{\nmk}$\\
  The prover wants to convince the verifier that he knows the solution $\bm{x}\in\Fqmn$ of the instance, i.e. such that $\bm{y} = \bm{H}\bm{x}$ \\[\baselineskip]
  \textbf{Step 1: Commitment} \\
  1. The prover builds a set of shares into a $(\ell+1,N)-$threshold secret sharing of $\bm{x}_B$. The prover also build shares of the annihilator polynomial $L=\prod_{u\in U}(X-u)=X^{q^r}+\sum_{i=0}^{r-1}\beta_iX^{q^i}$, a vector $\bm{a}$ uniformly sampled from $\Fqme^r$, and $c=-\langle\bm{\beta},\bm{a}\rangle$. Each party takes one share of the previous elements.\\
  For each $i\in\oneto{N}$ : $\state_i=(\share{\bm{x}_B}_i,\share{\bm{\beta}}_i,\share{\bm{a}}_i,\share{c}_i)$\\
  2. The prover computes the commitments : $\cmt_i=\commit{\state_i,\rho_i}$, where $\rho_i$ is a random tape sampled from $\{0, 1\}^\lambda$ for each $i\in\oneto{N}$.\\
  3. The prover computes and sends the Merkle tree root : $h_0 = \textsf{Merkle}(\cmt_1, \cdots, \cmt_N)$ \\[\baselineskip]
  \textbf{Step 2: First Challenge} \\
  4. The verifier sends $\big( (\gamma_j)_{j \in \oneto{n}}, \epsilon \big) \in \Fqme^n \times \Fqme$ to the prover.\\[\baselineskip]
  \textbf{Step 3: First Response} \\
 5. The prover chooses a public subset $S$ of parties such that $|S| = \ell + 1$: \\
 For each party $i\in S$ : \\
 \pcind \pcind - Compute $\share{\bm{x}_A}_i = \bm{y} - \bm{H}' \share{\bm{x}_B}_i$ and $\share{\bm{x}}_i = (\share{\bm{x}_A}_i \,||\, \share{\bm{x}_B}_i)$ \\
  \pcind \pcind - Compute $\share{z} = -\sum_{j=1}^n\gamma_j\left(\share{x_j}^{q^r}-\share{x_j}\right)$.\\
  \pcind \pcind - Compute  for each $k$ from 1 to $r-1$ : $\share{\omega_k} = \sum_{j=1}^n\gamma_j\left(\share{x_j}^{q^k}-\share{x_j}\right)$\\
 \pcind \pcind - Compute $\share{\boldsymbol{\alpha}}_i = \epsilon \cdot \share{\bm{w}}_i + \share{\bm{a}}_i$ and the parties reveals $\boldsymbol{\alpha}$\\
 \pcind \pcind -  Compute $\share{v}_i = \epsilon \cdot \share{z}_i - \dotp{\boldsymbol{\alpha}, \share{\bm{\beta}}_i} - \share{c}_i$ \\
 6. The prover commits the values : $h_1=\hash \left((\share{\boldsymbol{\alpha}}_i,\share{v}_i)_{i\in S}\right)$ \\ [\baselineskip]
  \textbf{Step 4: Second Challenge} \\
 7. The verifier sends a subset $ I \subset \oneto{N} $ of parties such that $ |I| = \ell $ to the prover \\[\baselineskip]
 \textbf{Step 5: Second Response}\\
 8. The prover sends $(\state_i,\rho_i)_{i\in I}$, and the authentication path to these commitments, to allow the verifier to rebuild the Merkle tree. The verifier deduces a value $\Tilde{h}_0$.\\
 9. The prover chooses a public party $i^*\in S\backslash I$, and sends $\share{\boldsymbol{\alpha}}_{i^*}$ to allow the verifier to recover $\share{\boldsymbol{\alpha}}=\reconstruct_{I\bigcup \{i^*\}}(\share{\boldsymbol{\alpha}}_{I\bigcup \{i^*\}})$\\
 10. The verifier sets $\share{v}_{i^*}$ such that $v=0$.\\
 11. The computes $\Tilde{h}_1=\hash((\share{\boldsymbol{\alpha}}_i,\share{v}_i)_{i\in S})$.\\
 12. The verifier outputs ACCEPT if $(\Tilde{h}_0,\Tilde{h}_1)=(h_0,h_1)$, REJECT otherwise.
}
\vspace{-\baselineskip}
\captionof{figure}{\footnotesize{Zero-knowledge protocol, deduced from the MPC protocol with threshold}}
\label{pzk_thr}
\end{figure}

%% file: fig-sign_thr.tex
\begin{figure}[h!]
    \pcb[codesize=\scriptsize, minlineheight=0.75\baselineskip, mode=text, width=0.98\textwidth] { 
  \textbf{Inputs} \\
  \pcind - Secret key $\sk = (\bm{x})$ with $\bm{x} = (\bm{x}_A \,||\, \bm{x}_B) \in \Fqmn$ such that $\rw{\bm{x}} = r$ \\
  \pcind - Public key $\pk = (\bm{H}, \bm{y})$ with $\bm{H} = (\bm{I} ~ \bm{H}') \in \Fqm^{\nmktn}$ and $\bm{y} \in \Fqm^{\nmk}$ such that $\bm{y} = \bm{H}\bm{x}$ \\
  \pcind - Message $m \in \{0, 1\}^*$ \\[\baselineskip]
  \textbf{Step 1: Commitment} \\
  1. Sample a random salt value $\salt \sampler \{0, 1\}^{2\lambda}$ \\
  2. Compute $\bm{\beta} = (\bm{\beta}_k)_{k \in [1, r - 1]} \in \Fqm^r$ the coefficients of the annihilator q-polynomial $L(X)$ associated to $\bm{x}$ such that $L(X)=\prod_{u\in U}(X-u) = X^{q^r} + \sum_{k = 1}^{r - 1} \bm{\beta}_k (X^{q^k}-X)$ and $\forall j \in \oneto{n}, L(\bm{x}_j) = 0$ \\
  3. For each iteration $e \in \oneto{\tau}$: \\
  \pcind $\diamond$ Sample $\bm{a}^{(e)} \sampler \Fqme^{r}$ and compute $c^{(e)}$ such that $c^{(e)} \in \Fqme$ and $c^{(e)} = - \dotp{\bm{\beta}, \bm{a}^{(e)}}$ \\
  \pcind $\diamond$ For each party $i \in \oneto{N}$: \\
  \pcind \pcind - Compute the $(\ell + 1, N)$-threshold LSSS $\share{\bm{x}_B^{(e)}}, \share{\bm{\beta}^{(e)}}, \share{\bm{a}^{(e)}}, \share{c^{(e)}}$ of $\bm{x}_B, \bm{\beta}, \bm{a}^{(e)}$ and $c^{(e)}$ \\
  \pcind \pcind - Compute $\state^{(e)}_i = (\share{\bm{x}^{(e)}_B}_i, \share{\bm{\beta}^{(e)}}_i, \share{\bm{a}^{(e)}}_i, \share{c^{(e)}}_i)$ and $\cmt^{(e)}_i = \hash_0(\salt, e, i, \state^{(e)}_i)$ \\
  \pcind $\diamond$ Compute the Merkle tree root $h^{(e)}_0 = \textsf{Merkle}(\cmt^{(e)}_1, \cdots, \cmt^{(e)}_N)$ \\
  4. Compute $h_1 = \hash_1(m, \pk, \salt, (h^{(e)}_0)_{e \in \oneto{\tau}})$ \\[\baselineskip]
  \textbf{Step 2: First Challenge} \\
  5. Extend hash $\big( (\gamma^{(e)}_j)_{j \in \oneto{n}}, \epsilon^{(e)} \big)_{e \in \oneto{\tau}} \leftarrow PRG(h_1)$ where $\big( (\gamma^{(e)}_j)_{j \in \oneto{n}}, \epsilon^{(e)} \big)_{e \in \oneto{\tau}}\in (\Fqme^n \times \Fqme)^{\tau}$ \\[\baselineskip]
  \textbf{Step 3: First Response} \\
  6. For each iteration $e \in \oneto{\tau}$: \\
  \pcind $\diamond$ For each party $i \in S$ with $S$ a public subset of parties such that $|S| = \ell + 1$: \\
  \pcind \pcind - Compute $\share{\bm{x}_A^{(e)}}_i = \bm{y} - \bm{H}' \share{\bm{x}_B^{(e)}}_i$ and $\share{\bm{x}^{(e)}}_i = (\share{\bm{x}_A^{(e)}}_i \,||\, \share{\bm{x}_B^{(e)}}_i)$ \\
  \pcind \pcind - Compute $\share{z^{(e)}}_i = - \sum\nolimits_{j = 1}^{n} \gamma_j \left(\share{x_j^{(e)}}_i^{q^r}-\share{x_j^{(e)}}_i\right)$ and $\forall k \in [1, r - 1], \share{w_k^{(e)}}_i = \sum\nolimits_{j = 1}^{n} \left(\gamma_j \share{x_j^{(e)}}_i^{q^k}-\share{x_j^{(e)}}_i\right)$ \\
  \pcind \pcind - Compute $\share{\boldsymbol{\alpha}^{(e)}}_i = \epsilon^{(e)} \cdot \share{\bm{w}^{(e)}}_i + \share{\bm{a}^{(e)}}_i$ \\
  \pcind \pcind - Compute  $\share{v^{(e)}}_i = \epsilon^{(e)} \cdot \share{z^{(e)}}_i - \dotp{\boldsymbol{\alpha}^{(e)}, \share{\bm{\beta}^{(e)}}_i} - \share{c^{(e)}}_i$ \\
  7. Compute $h_2 = \hash_2(m, \pk, \salt, h_1, (\share{\boldsymbol{\alpha}^{(e)}}_i, \share{v^{(e)}}_i)_{i \in S, e \in \oneto{\tau}})$ \\[\baselineskip]
  \textbf{Step 4: Second Challenge} \\
  8. Extend hash $(I^{(e)})_{e \in \oneto{\tau}} \leftarrow PRG(h_2)$ where $ (\{ I \subset N ~|~ |I| = \ell \})^{\tau}$ \\[\baselineskip]
  \textbf{Step 5: Second Response} \\
  9. For each iteration $e \in \oneto{\tau}$: \\
  \pcind $\diamond$ Choose deterministically a party $i^{*(e)} \in S \, \backslash \, I^{(e)}$ and compute $\share{\boldsymbol{\alpha}^{(e)}}_{i^{*(e)}}$ \\
  \pcind $\diamond$ Compute the authentication path $\mathsf{auth}^{(e)}$ associated to root $h^{(e)}_0$ and $(\cmt^{(e)}_i)_{i \in I^{(e)}}$  \\
  \pcind $\diamond$ Compute $\rsp^{(e)} = \big( (\share{\bm{x}^{(e)}_B}_i, \share{\bm{\beta}^{(e)}}_i, \share{\bm{a}^{(e)}}_i, \share{c^{(e)}}_i)_{i \in I^{(e)}}, \mathsf{auth}^{(e)}, \share{\boldsymbol{\alpha}^{(e)}}_{i^{*(e)}} \big)$ \\
  10. Compute $\sigma = (\salt, h_1, h_2, (\rsp^{(e)})_{e \in \oneto{\tau}})$
}
\vspace{-\baselineskip}
\captionof{figure}{\footnotesize{Rank based signature scheme with threshold - Signing algorithm}}
\label{sign_thr}
\end{figure}

%% file: fig-verify_thr.tex
\begin{figure}[h!]
    \pcb[codesize=\scriptsize, minlineheight=0.75\baselineskip, mode=text, width=0.98\textwidth] { 
  \textbf{Inputs} \\
  \pcind - Public key $\pk = (\bm{H}, \bm{y})$ with $\bm{H} = (\bm{I} ~ \bm{H}') \in \Fqm^{\nmktn}$ and $\bm{y} \in \Fqm^{\nmk}$ such that $\bm{y} = \bm{H}\bm{x}$ \\
  \pcind - Message $m \in \{0, 1\}^*$ \\
  \pcind - Signature $\sigma = (\salt, h_1, h_2, (\bar{\rsp}^{(e)})_{e \in \oneto{\tau}})$ \\[\baselineskip]
  \textbf{Step 1: Parse signature} \\
  1. Sample $\big( (\gamma^{(e)}_j)_{j \in \oneto{n}}, \epsilon^{(e)} \big)_{e \in \oneto{\tau}} \samples{h_1} (\Fqme^n \times \Fqme)^{\tau}$ \\
  2. Sample $(I^{(e)})_{e \in \oneto{\tau}} \samples{h_2} (\{ I \subset N ~|~ |I| = \ell \})^{\tau}$ \\
  3. For each iteration $e \in \oneto{\tau}$: \\
  \pcind $\diamond$ Choose deterministically $i^{*(e)}$ from $S\setminus I^{(e)}$ \\
  \pcind $\diamond$ Parse $\bar{\rsp}^{(e)} := \big( (\share{\bar{\bm{x}}^{(e)}_B}_i, \share{\bar{\bm{\beta}}^{(e)}}_i, \share{\bar{\bm{a}}^{(e)}}_i, \share{\bar{c}^{(e)}}_i)_{i \in \bar{I}^{(e)}}, \bar{\mathsf{auth}}^{(e)}, \share{\bar{\boldsymbol{\alpha}}^{(e)}}_{i^{*(e)}} \big)$ \\[\baselineskip]
  \textbf{Step 1: Recompute $\bm{h}_1$} \\
  4. For each iteration $e \in \oneto{\tau}$: \\
  \pcind $\diamond$ For each party $i \in I^{(e)}$: \\
  \pcind \pcind - Compute $\bar{\state}^{(e)}_i = (\share{\bar{\bm{x}}^{(e)}_B}_i, \share{\bar{\bm{\beta}}^{(e)}}_i, \share{\bar{\bm{a}}^{(e)}}_i, \share{\bar{c}^{(e)}}_i)$ and $\bar{\cmt}^{(e)}_i = \hash_0(\salt, e, i, \bar{\state}^{(e)}_i)$ \\
  \pcind $\diamond$ Compute the Merkle tree root $\bar{h}^{(e)}_0$ from $(\bar{\cmt}^{(e)}_i)_{i \in I^{(e)}}$ and $\bar{\mathsf{auth}}^{(e)}$ \\
  5. Compute $\bar{h}_1 = \hash_1(m, \pk, \salt, (\bar{h}^{(e)}_0)_{e \in \oneto{\tau}})$ \\[\baselineskip]
  \textbf{Step 2: Recompute $\bm{h}_2$} \\
  6. For each iteration $e \in \oneto{\tau}$: \\
  \pcind $\diamond$ For each party $i \in I^{(e)}$: \\
  \pcind \pcind - Compute $\share{\bar{\bm{x}}_A^{(e)}}_i = \bm{y} - \bm{H}' \share{\bar{\bm{x}}_B^{(e)}}_i$ and $\share{\bar{\bm{x}}^{(e)}}_i = (\share{\bar{\bm{x}}_A^{(e)}}_i \,||\, \share{\bar{\bm{x}}_B^{(e)}}_i)$ \\
  \pcind \pcind - Compute $\share{\bar{z}^{(e)}}_i = - \sum\nolimits_{j = 1}^{n} \gamma_j\left( \share{\bar{x}_j^{(e)}}_i^{q^r}-\share{\bar{x}_j^{(e)}}_i\right)$ and $\forall k \in [1, r - 1], \share{\bar{w}_k^{(e)}}_i = \sum\nolimits_{j = 1}^{n} \gamma_j \left(\share{\bar{x}_j^{(e)}}_i^{q^k}-\share{\bar{x}_j^{(e)}}_i\right)$ \\
  \pcind \pcind - Compute $\share{\bar{\boldsymbol{\alpha}}^{(e)}}_i = \epsilon^{(e)} \cdot \share{\bar{\bm{w}}^{(e)}}_i + \share{\bar{\bm{a}}^{(e)}}_i$ \\
  \pcind $\diamond$ Reconstruct $\bar{\boldsymbol{\alpha}}^{(e)}$ and $(\share{\bar{\boldsymbol{\alpha}}^{(e)}}_i)_{i \in S}$ from $(\share{\bar{\boldsymbol{\alpha}}^{(e)}}_i)_{i \in I^{(e)}}$ and $\share{\bar{\boldsymbol{\alpha}}^{(e)}}_{i^{*(e)}}$ \\
  7. For each iteration $e \in \oneto{\tau}$: \\
  \pcind $\diamond$ For each party $i \in I^{(e)}$: \\
  \pcind \pcind - Compute $\share{\bar{v}^{(e)}}_i = \epsilon^{(e)} \cdot \share{\bar{z}^{(e)}}_i - \dotp{\bar{\boldsymbol{\alpha}}^{(e)}, \share{\bar{\bm{\beta}}^{(e)}}_i} - \share{\bar{c}^{(e)}}_i$ \\
  \pcind $\diamond$ Reconstruct $(\share{\bar{v}^{(e)}}_i)_{i \in S}$ from $(\share{\bar{v}^{(e)}}_i)_{i \in I^{(e)}}$ and $\bar{v}^{(e)} = 0$ \\
  8. Compute $\bar{h}_2 = \hash_2(m, \pk, \salt, \bar{h}_1, (\share{\bar{\boldsymbol{\alpha}}^{(e)}}_i, \share{\bar{v}^{(e)}}_i)_{i \in S, e \in \oneto{\tau}})$ \\[\baselineskip]
  \textbf{Step 5: Verify signature} \\
  9. Return $(\bar{h}_1 = h_1) \wedge (\bar{h}_2 = h_2)$
}
\vspace{-\baselineskip}
\captionof{figure}{\footnotesize{Rank based signature scheme with threshold - Verification algorithm}}
\label{ver_thr}
\end{figure}

%% file: param-hyp.tex
\begin{table}[H]
\begin{center}
{\setlength{\tabcolsep}{0.4em}
{\renewcommand{\arraystretch}{1.6}
{\scriptsize
  \begin{tabular}{|c||c|c|c|c|c||c|c|c|c||c|c|}
    \hline
    NIST security level & $q$ & $m$ & $n$ & $k$ & $r$ & $N$ & $\eta$ & $\tau$ & $\pk$ & $\sigma$ \\ \hline
    
    1 & 2 & 31 & 33 & 15 & 10 & 256  & 1 & 20 & 0.1 kB & 5.9 kB \\ \hline
    3 & 2 & 37 & 41 & 18 & 13 & 256  & 1 & 29 & 0.1 kB & 12.9 kB \\ \hline
    5 & 2 & 43 & 47 & 18 & 17 & 256  & 1 & 38 & 0.1 kB & 22.8 kB \\ \hline
  \end{tabular}
  \vspace{0.5\baselineskip}
  \caption{Parameters for additive MPC on hypercube, short signature}
  \label{table_add_short}
}}}
\end{center}
\end{table}

\begin{table}[H]
\begin{center}
{\setlength{\tabcolsep}{0.4em}
{\renewcommand{\arraystretch}{1.6}
{\scriptsize
  \begin{tabular}{|c||c|c|c|c|c||c|c|c|c||c|c|}
    \hline
    NIST security level & $q$ & $m$ & $n$ & $k$ & $r$ & $N$ & $\eta$ & $\tau$ & $\pk$ & $\sigma$ \\ \hline
    
    1 & 2 & 31 & 33 & 15 & 10 & 32  & 1 & 30 & 0.1 kB & 7.4 kB \\ \hline
    3 & 2 & 37 & 41 & 18 & 13 & 32  & 1 & 44 & 0.1 kB & 16.4 kB \\ \hline
    5 & 2 & 43 & 47 & 18 & 17 & 32  & 1 & 58 & 0.1 kB & 29.1 kB \\ \hline
  \end{tabular}
  \vspace{0.5\baselineskip}
  \caption{Parameters for additive MPC on hypercube, fast signature}
  \label{table_add_fast}
}}}
\end{center}
\end{table}

%% file: param-l3.tex
\begin{table}[H]
\begin{center}
{\renewcommand{\arraystretch}{1.6}
{\setlength{\tabcolsep}{0.4em}
{\scriptsize
  \begin{tabular}{|c||c|c|c|c|c||c|c|c|c||c|c|}
    \hline
    NIST security level & $q$ & $m$ & $n$ & $k$ & $r$ & $N$ & $\eta$ & $\tau$ & $\pk$ & $\sigma$ \\ \hline
    
    1 & 256 & 11 & 12 & 5 & 5 & 256  & 2 & 6 & 0.1 kB & 8.2 kB \\ \hline
    3 & 256 & 13 & 17 & 7 & 6 & 256  & 1 & 11 & 0.1 kB & 18.3 kB \\ \hline
    5 & 256 & 17 & 17 & 7 & 7 & 256  & 3 & 14 & 0.1 kB & 32.5 kB \\ \hline
  \end{tabular}
  \vspace{0.5\baselineskip}
  \caption{Parameters for threshold with $\ell=3$, $q=256$}
  \label{table_thr_l3}
}}}
\end{center}
\end{table}

\begin{table}[H]
\begin{center}
{\renewcommand{\arraystretch}{1.6}
{\setlength{\tabcolsep}{0.4em}
{\scriptsize
  \begin{tabular}{|c||c|c|c|c|c||c|c|c|c||c|c|}
    \hline
    NIST security level & $q$ & $m$ & $n$ & $k$ & $r$ & $N$ & $\eta$ & $\tau$ & $\pk$ & $\sigma$ \\ \hline
    
    1 & 2 & 31 & 33 & 15 & 10 & 256  & 2 & 18 & 0.1 kB & 9.3 kB \\ \hline
    3 & 2 & 37 & 41 & 18 & 13 & 256  & 2 & 27 & 0.1 kB & 21.4 kB \\ \hline
    5 & 2 & 43 & 47 & 18 & 17 & 256  & 2 & 35 & 0.1 kB & 34.8 kB \\ \hline
  \end{tabular}
  \vspace{0.5\baselineskip}
  \caption{Parameters for threshold with $\ell=1$, $q=2$}
  \label{table_thr_l1}
}}}
\end{center}
\end{table}

%% file: fig-hvzk_hyp.tex
\begin{figure}[h!]
    \pcb[codesize=\scriptsize, minlineheight=0.75\baselineskip, mode=text, width=0.98\textwidth] { 
   \textbf{Public data} : An instance of a Rank-SD problem : $\bm{H} = (I ~ \bm{H}') \in \Fqm^{\nmktn}$ and $\bm{y}\in \Fqm^{\nmk}$\\
   [\baselineskip]
  \textbf{Step 1: Sample challenges} \\
  1. Sample challenges :\\
  \pcind \pcind - First challenge : $\ch_1=\big( (\gamma_j)_{j \in \oneto{n}}, \epsilon \big)\sampler \Fqme^n \times \Fqme$ \\
  \pcind \pcind - Second challenge : $\ch_2= i^* \sampler\oneto{N}$ \\
  [\baselineskip]
  \textbf{Step 2: Compute shares and their commitments} \\
  2. Sample a seed for pseudo-random generator : $\seed\sampler\{0, 1\}^\lambda$\\
  3. Expand root seed recursively using TreePRG to obtain $N$ leafs and seeds $(\seed_{i},\rho_{i})$.\\
  4. For each $i\in\oneto{N}\backslash\lbrace i^*\rbrace$ : \\
  \pcind \pcind - Sample $\share{\bm{a}}_i\samples{\seed_i} \Fqme^{r-1}$\\
  \pcind \pcind If $i\neq N$\\
  \pcind \pcind \pcind \pcind  - Sample $(\share{\bm{x}_B}_i,\share{\bm{\beta}}_i,\share{c}_i)\samples{\seed_i} TPRG$\\
  \pcind \pcind \pcind \pcind - $\state_i=\seed_i$\\
  \pcind \pcind - Else:\\
  \pcind \pcind \pcind \pcind - Sample $\left(\share{\bm{x}_B}_{N},\share{\bm{\beta}}_{N},\share{c}_{N}\right)\sampler \Fqmk\times\Fqm^r\times\Fqme$\\
    \pcind \pcind \pcind \pcind - $\aux_N=\left(\share{\bm{x}_B}_{N},\share{\bm{\beta}}_{N},\share{c}_{N}\right)$\\
  \pcind \pcind \pcind \pcind - $\state_{N}=\seed_N\,||\,\aux_N$\\
  \pcind \pcind - Simulate the computation of the party $i$ to get $\share{\bm{\alpha}}_i$ and $\share{v}_i$.\\
  5. For the party $i^*$:\\
  \pcind \pcind - $\share{\bm{\alpha}}_{i^*}\sampler\Fqme^{r-1}$\\
  \pcind \pcind - $\share{v}_{i^*}=-\sum_{i\neq i^*}\share{v}_i$\\
  6. Compute the commitment: $\cmt_{i^*}=\commit{\state_{i^*},\rho_{i^*}}$.\\
  7. Commit the full hypercube: $h_0=\commit{cmt_1,...,cmt_N}$.\\
  [\baselineskip]
  8. For each main party $p\in\oneto{2}$ : compute $\share{\bm{\alpha}}_p$ and $\share{v}_p$\\
  9. For each dimension $k\in\oneto{D}$ : compute $H_k=\hash (\share{\bm{\alpha}},\share{v})$\\
  10. Compute $h_1=\hash (H_1,...,H_D)$\\
  [\baselineskip]
  \textbf{Step 3 : Output transcript} \\
  11. The prover outputs the transcript $(h_0,\ch_1,\rsp_1,\ch_2,\rsp_2)$, where $\rsp_1=h_1$ and $\rsp_2=\left(\cmt_{i^*},\share{\bm{\alpha}}_{i^*},(\state_j,\rho_j)_{j\neq i^*}\right)$.
}
\vspace{-\baselineskip}
\captionof{figure}{\footnotesize{HVZK simulator for zero-knowledge proof optimized with the hypercube technique}}
\label{hvzk_hyp}
\end{figure}